\documentclass[11pt,a4paper]{article}
\usepackage[margin=1.0 in]{geometry}
\usepackage{float}
\usepackage{graphicx} 
\usepackage{hyperref}
\usepackage{xcolor}
\usepackage{xurl}
\usepackage{amsfonts}
\usepackage{amsmath}
\usepackage{amsthm}
\usepackage{amssymb}
\newcommand{\ignore}[1]{}

\newtheorem{theorem}{Theorem}
\newtheorem*{theorem*}{Theorem}
\newtheorem{definition}{Definition}

\usepackage{boxedminipage}
\usepackage{graphicx}
\usepackage{mdframed}
\usepackage{float}

\newcommand{\JM}[1]{\textcolor{red}{\textbf{Julian: }{#1}}}
\newenvironment{thickframe}{\begin{center}\begin{boxedminipage}{0.98\textwidth}}{\end{boxedminipage}\end{center}}
\title{Multiple Proposer Transaction Fee Mechanism Design: Robust Incentives Against Censorship and Bribery}
\author{
  Aikaterini-Panagiota Stouka\\
  Nethermind Research\\
  \href{mailto:aikaterini-panagiota@nethermind.io}{aikaterini-panagiota@nethermind.io}\\
  \and
  Julian Ma\\
  Ethereum Foundation \\
  \href{mailto:julian@ethereum.org}{julian@ethereum.org} \\
  \and
  Thomas Thiery \\
  Ethereum Foundation \\
  \href{mailto:thomas@ethereum.org}{thomas@ethereum.org} \\
}
\date{March 2025}

\begin{document}

\maketitle

\begin{abstract}
    Transaction Fee Mechanism (TFM) design in blockchain protocols has gained significant attention following the pioneering work of Roughgarden [EC '21], which established a formal framework for analyzing user and block proposer incentives under various Transaction Fee Mechanisms, including Ethereum’s current fee mechanism EIP-1559. However, the original TFM framework and follow-up TFM works overlook the critical challenge of censorship resistance—specifically in the presence of an external malicious actor who is willing to bribe the proposer to censor a transaction. In this paper, we extend the Roughgarden's framework to capture censorship resistance under bribery attacks via a Bayesian game, where a strategic block proposer's "type" is determined by a bribe function from an external malicious actor. Under this framework, the definition of a standard TFM is extended to a bribery-aware TFM. This technique is broadly applicable to analyze censorship resistance under bribery attacks of both single and multiple proposer protocols within the original TFM scope. We choose to utilize it to evaluate the incentive compatibility and censorship resistance of several TFMs within the context of a multiple proposer protocol called Fork-Choice Enforced Inclusion Lists (FOCIL). FOCIL represents a critical evolution in the Ethereum roadmap, serving as the consensus and censorship resistance flagship for the upcoming Hegotá hard fork. It aims to bolster Ethereum's censorship resistance by enabling multiple proposers to contribute to block construction. While recent works such as Garimidi et al.[FC'25] have extended the TFM framework to multiple proposer settings, they do not aim to capture censorship under bribery attacks and they are not compatible with the unique hierarchical structure of FOCIL.
\end{abstract}

\section{Introduction}

Censorship resistance serves as a fundamental tenet of blockchain design. Currently, most protocols elect a single party per slot-the proposer\footnote{In this paper, we use Ethereum's terminology to reference consensus parties. Other blockchains may use different terminology for similar roles.}-to construct a block from user transactions broadcast across the peer-to-peer network. This collection of pending transactions is known as the mempool. In contemporary architectures, proposers often delegate the task of block construction to specialized entities, known as builders, in exchange for a tip; this process is typically facilitated through external auctions (cf.\cite{PBS}). Empirical evidence from studies such as \cite{heimbach2023pbsreality} and \cite{10.1145/3589334.3645431} confirms that Ethereum proposers (or builders) engage in active censorship in practice.

\paragraph{Transaction Fee Mechanism design and overview of \cite{roughgarden2021tfm}}
Transaction Fee Mechanism (TFM) design, pioneered by Roughgarden \cite{roughgarden2021tfm} and further developed by \cite{Bahrani2024PostMEV, Wu2024Maximizing, Gafni2025Welfare, chung2022foundations, Cullen2024APP},
establishes a formal framework for analyzing the incentives of users and proposers under fee mechanisms that determine how transaction fees (tips that users are willing to pay for their transaction inclusion) are allocated within the protocol.

In \cite{roughgarden2021tfm}, a single party proposes blocks; we refer to it as block producer. Users submit transactions to a mempool observed by this block producer, holding private valuations for their transaction's inclusion \footnote{It is assumed that the position of a transaction within a block does not affect a user's valuation and that users cannot observe other transactions in the mempool prior to submission.}. For each transaction, the sender specifies a fee that is paid only upon inclusion. The block producer chooses an \textit{allocation rule} $x$ to determine which user transactions are included in their block. Additionally, the block producer can insert their own transactions, referred to as ``fake" transactions, into their block. \par 
In this setting, a user's strategy is the fee they choose (referred to as their \textit{bid}), while the block producer's strategy consists of their allocation rule and their set of fake transactions. For each included user transaction, the block producer receives a portion of the fee defined by a\textit{ payment rule} $p$. A \textit{burning rule} $q$ determines the portion of the fee that is "burned" or redirected to an entity different from the block producer (for instance, the blockchain protocol's treasury). The payment and the burning rule are determined by the underlying blockchain protocol. Note that when a block producer adds fake transactions, they effectively only pay the fee that is going to be burnt, as the remainder of the fee is returned to themselves. A TFM is formally defined by the triple $(x,p,q)$. In \cite{roughgarden2021tfm}, the following properties related to the incentives of the users and the block producer are introduced.
\begin{itemize} 
\item A TFM $(x,p,q)$ is Dominant-Strategy Incentive Compatible (DSIC) when assuming that
the block producer follows the allocation rule $x$, every user has a dominant strategy, no
matter the private valuations or the bids of the other users.
\item  A TFM $(x,p,q)$ is Incentive Compatible for  Myopic Miners (MMIC), if for every
history $H$ of blocks in the blockchain and mempool $M$, the block producer maximizes their utility if: (i) they follow the allocation rule x and (ii) they do not add any fake transactions. The term “myopic” implies that the block producer tries to maximize their profit from the current block.
\item Another property introduced in \cite{roughgarden2021tfm}  is Off-Chain Agreement-proof (OCA-proof). At a high level, this property examines whether a coalition between the block producer and the
users can increase all of their utilities through off-chain payments.
\end{itemize}

While the current TFM literature has been instrumental in characterizing the economic stability of blockchain protocols, to the best of our knowledge, it does not provide a framework that simultaneously captures both the key properties of the framework introduced in \cite{roughgarden2021tfm}  and censorship resistance.

\paragraph{Multiple Proposer Architectures}
To bolster censorship resistance, protocol designers have proposed mechanisms that distribute authority across multiple parties per block. One example is Multiple Concurrent Block Proposers design (cf. \cite{garimidi2025multipleconcurrentproposers}) that involves multiple
proposers whose input is ordered by a predetermined ordering rule. Alternatively, inclusion lists, introduced in \cite{neuder2023eip7547} and later expanded as Fork-Choice Enforced Inclusion Lists (FOCIL) \cite{thiery2024eip7805}, allow an inclusion list committee to impose constraints on transaction inclusion, while final ordering is determined by a single lead proposer. Notably, FOCIL represents a critical evolution in Ethereum’s roadmap, as it serves as the consensus and censorship headliner for the upcoming Hegota hard fork \cite{consensusheadliner}, which is scheduled to follow the Pectra and Glamsterdam upgrades in late 2026.

At present, FOCIL design relies primarily on altruistic assumptions; it lacks a natively integrated TFM and has not been evaluated using the rigorous game-theoretic models established by Roughgarden\cite{roughgarden2021tfm} . While some recent studies have extended TFM design to multiple proposer settings, they typically assume symmetric power among all proposers. Such models fail to capture the unique asymmetric dynamics inherent in the FOCIL architecture, where a single lead proposer retains final ordering authority but is constrained by an inclusion list committee. In this paper, the lead proposer is referred to as the block producer (as in single proposer protocols), while the inclusion list committee is referred to as includers.

\subsection{Our results}

Our contribution is twofold. First, we extend Transaction Fee Mechanism design to incorporate censorship resistance under bribery attacks---a contribution of independent interest for any blockchain protocol that can be studied using Roughgarden's model \cite{roughgarden2021tfm} or within TFM scope. Second, we apply and further extend this framework to study FOCIL \cite{thiery2024eip7805}, the censorship resistance headliner of Ethereum's upcoming Hegota hard fork. In more detail, (i) we address a critical question left underexplored in FOCIL: how transaction fees should be distributed among multiple proposers to enhance censorship resistance when includers are rational and not altruistic and (ii) we prove that FOCIL maintains key incentive compatibility guarantees of the Ethereum TFM (cf. \cite{roughgarden2021tfm}) across both the current and our proposed TFMs. We note that off-chain agreements between the involved parties are outside the scope of this paper. Extending our model to capture such agreements---for instance, by enhancing the OCA-proof property---is left for future work.

\paragraph{Overview: extending TFM to account for censorship resistance under bribery attacks}
\label{briberyextension}
At a high level, we model bribery from an external malicious briber via an \textit{interim stage Bayesian game} in which the strategic type of the proposer(s) is determined by a bribe function that belongs to a set of candidate bribe functions. A bribe function determines the final reward/bribe the proposer receives from a malicious external briber who wants to censor a transaction. This function considers several factors such as the transaction characteristics and the proposers' strategy. Users are uncertain about the bribe functions of the proposer(s) and hold beliefs about them. We define this extended TFM as Bribery-aware TFM. A Bribery-aware TFM will be \textit{censorship resistant} under a set of bribe functions and users' beliefs if its allocation rule is bribe agnostic (does not account for bribes). 
 
\paragraph{Overview: extending \cite{roughgarden2021tfm} further to capture FOCIL}
\label{appendixcompariso}
 Although the technique of determining a proposer's strategic type via a bribe function can be applied to analyze censorship resistance of both single and multiple proposer protocols within the TFM framework, we specifically extend \cite{roughgarden2021tfm} to capture the asymmetric multiple proposer case that is compatible with FOCIL.
In more detail, our model expands\cite{roughgarden2021tfm} by incorporating several new dimensions. \par In more detail, we further account for the incentives of the includers. Like the block producer, includers possess strategic types determined by their respective bribe functions. In this incomplete information setting, each includer is aware of their own type but remains uncertain regarding the types of other includers and the block producer, maintaining only beliefs over these external types. Moreover, we extend the model to capture the fact that in FOCIL, a block that does not meet the inclusion lists requirements is rejected by attesters, entities on Ethereum that vote for the validity of the block. We abstract the role of attesters by assuming that a block gives zero transaction fees to the block producer, if it fails to satisfy inclusion list requirements. \par 
Furthermore, we consider additional strategies available to the block producer in the context of censorship, particularly under bribery attacks, and when it needs to satisfy  inclusion list requirements specified by the FOCIL protocol. In Roughgarden’s model, the block producer selects the allocation rule and may include fake transactions---but only directly in their block, not in the mempool. When a single party controls the block inclusion, and the users send transactions before they observe the mempool, there is no practical difference between the block producer (i) adding a transaction to the mempool and later including it in the block, and (ii) directly inserting the transaction into the block. In our setting, however, multiple parties independently contribute to inclusion lists or block construction. As a result, injecting fake transactions into the mempool becomes a strategically distinct action from directly including them in the inclusion list or the block. This distinction opens up new censorship-related strategies for the includers and the block producer. \par In addition, we account for the possibility that the block producer may be uncertain about whether they will actually create the block in the current slot---for instance, if the block producer is acting as a builder \cite{PBS}. To model this uncertainty in a single-shot game, based on the $\gamma$-strict utility of \cite{chung2022foundations}, we assume that if the block producer injects fake transactions into the mempool, they incur a cost; they pay a fraction $\gamma$ of the fake transaction’s fee. This parameter reflects the likelihood that the block producer will not end up creating the block and receive back the fee of these transactions that would go to the party who would create the block. \par Finally, in addition to the censorship resistance property described in the previous paragraph, we introduce a property called fair-under-congestion. At a high level, this property is satisfied by a TFM if the following holds: for every mempool and every transaction, there exists a transaction fee that ensures the transaction is included in both an inclusion list and the block, assuming parties follow the allocation rules specified by the TFM. This property is particularly critical for protocols such as FOCIL, where during periods of congestion (when the volume of pending transactions exceeds available block space), the block producer is permitted to omit transactions from inclusion lists without incurring a penalty. For instance, a payment rule that allocates the entire transaction fee to includers disincentivizes the block producer to fill their block with user transactions during periods of congestion.

\section{Other Related Works}
Apart from the related literature already discussed, our work is perhaps most closely related to \cite{Cullen2024APP}, which extends Roughgarden's model to multiple concurrent proposers. We differentiate our work by adhering to the EIP-1559 burning rule and focusing on a censorship resistant in the presence of a bribing adversary. Furthermore, we capture blockchain protocols with multiple proposers who do not propose blocks concurrently, but instead fulfill different roles sequentially. This approach captures FOCIL \cite{thiery2024eip7805}, which is the consensus and censorship resistant headliner for Ethereum's upcoming hard fork Hegota \cite{consensusheadliner}. The design of Transaction Fee Mechanisms (TFMs) for multiple proposers has also been explored in  \cite{garimidi2025transactionfeemechanismdesign,10.1007/978-3-032-07035-7_2}. However, our work differs in two key aspects. First, their research focuses on DAG-based protocols where multiple proposers submit blocks for the same height (distance from the first block) simultaneously, relying on an ordering rule to sequence transactions. In contrast, FOCIL involves multiple includers who propose inclusion lists, after which a single block producer selects which of these transactions are finalized in the block for that height. Second, their model does not account for bribery attacks or censorship resistance, which are the primary focus of this paper. \par
We note that our model is more complex than the unit-demand auction design (i.e., selling multiple items where each buyer desires at most one) \cite{nisan2007algorithmic} because in our case, multiple entities simultaneously ``sell" space for transactions and their available space is interdependent. For example, if an includer includes a transaction, the block producer incurs a loss by omitting it (e.g., as we will explain in the next section, their block is disregarded by the attesters, or they need to pay to fill their block with their own transactions). Conversely, even if an includer includes a transaction, its inclusion in the final block is not guaranteed. \par
We use modeling techniques from the literature studying censorship resistance. In \cite{10.1145/3658644.3670330} the authors study bribery attacks on consensus protocols. We adopt the functionality proposed in that work where bribes may depend not only on the strategies of individuals but also on the strategies of other parties. In \cite{cryptoeprint:2025/194} the authors propose AUCIL, an inclusion list design, and quantify its censorship resistance. Instead, we focus in FOCIL, and we study whether inclusion lists affect Ethereum's existing TFM. Moreover, our modeling techniques differ since we use an incomplete information setting (where proposers do not know the bribe function of the other proposers) and do not assume a coordination device, used in \cite{cryptoeprint:2025/194} to obtain a correlated Nash equilibrium. \cite{berger2025fraudproofs} studies economic censorship games but focuses on fraud proofs in optimistic rollups. 
Finally, \cite{fox2023censorshipresistanceonchainauctions} studies the censorship resistance of blockchain applications under a bribing, rational adversary and proposes a payment rule for multiple concurrent proposers that increases the cost of censorship. In contrast, our work examines censorship resistance through the lens of TFM design, enabling the joint evaluation of both a protocol’s resilience to bribery-based censorship, as well as the incentive compatibility of all participating parties. Moreover, our analysis focuses on FOCIL, which differs fundamentally from multiple concurrent proposer protocols. Furthermore, in our case, the briber is (i) malicious rather than rational, and may choose to bribe irrespective of their own value; (ii) is aware of users' transactions before they decide to bribe.(e.g., to target a specific competitor); and (iii) is more sophisticated, with the ability to offer conditional bribes.

\section {Background}
Below we give an overview of FOCIL \cite{thiery2024eip7805}, the consensus and censorship resistant headliner for Ethereum's upcoming hard fork Hegota \cite{consensusheadliner}.
FOCIL is an inclusion list design in which each member of a committee can submit a list of transactions, called \textit{inclusion list}, that must be included in the block for the current slot. This block is created by a different party, referred to as the \textit{block producer}. The committee members are known as \textit{includers} (cf.\cite{includers}). The block producer may add additional transactions not present in any inclusion list, but they must adhere to the rules set by the FOCIL protocol regarding the inclusion of list transactions. In particular, if the block is not fully filled with transactions, the producer is required to use transactions from the inclusion lists to occupy the remaining space unless these transactions are invalid (according to Ethereum's rule for transaction validity). 
\par 
When includers prepare their inclusion lists, they propagate them not only to the block producer but also to the attesters who will verify that the block producer has followed the inclusion list requirements. Any block that, in the view of the attesters, violates these rules is disregarded. 
For more details on the protocol design, cf. \cite{thiery2024eip7805}. Moreover, \cite{ma2024uncrowdable} explores the design rationale behind FOCIL, specifically why the block producer is still given monopoly ordering rights. \\
In our work, we analyze not only FOCIL as described above but also two different candidate implementations that affect its game-theoretic analysis: (i) \textit{unconditional FOCIL}: inclusion list transactions must be in the block even if the block is full and (ii) \textit{unique senders FOCIL}: each includer may only include one transaction per sender in their list. Moreover, we assume that includers can be ordered according to a deterministic algorithm. 

\section{Bribery-aware Transaction Fee Mechanism Design}
\label{briberyaware}
We model bribery-based censorship in single proposer TFM design as follows:
A malicious briber may pay the block producer to omit (censor) specific transactions. To incorporate this in Roughgarden’s model \cite{roughgarden2021tfm}, we assign the block producer a type (potentially known only to themselves), corresponding to a bribe function. This function specifies how the briber compensates the party based on both their strategy and the strategy of other parties (we extend bribing attacks presented in \cite{10.1145/3658644.3670330}). Then, Roughgarden’s model \cite{roughgarden2021tfm} is extended as follows: 

\paragraph{Types of the user}
Every user is of a specific type determined by the characteristics of the transaction they want to send. Every transaction has the following characteristics: (i)  \textit{size} $s_t$; (ii) \textit{value} $v_t$ per unit of size; (iii) the bid $b_{t}$  per unit of size and (iv)  \textit{public information} $p_{t}$. The value $v_t$ reflects the maximum amount per unit of size the user is willing to pay for the transaction to be included in the block. We follow \cite{roughgarden2021tfm} and assume that the value is not affected by the position of the transaction in the block. The value is known only to the user. The bid $b_{t}$ is the fee the user will pay if their transaction is included in the block. The public information $p_{t}$ is metadata of the transaction, such as its sender. This characteristic is provided as input to the bribe function to account for situations where the briber intends to censor transactions originating from a specific user. 
\paragraph{Beliefs of the user}
Every user $i$ has a belief regarding (i) the type of every other user, denoted by $\text{Belief}_{\text{User}_i \rightarrow \text{Users}}$, and (ii) the type of the block producer, denoted by $\text{Belief}_{\text{User}_i \rightarrow BP}$.

\paragraph{Types of the block producer}
The block producer's type is determined by the bribe they receive from the malicious briber for a target transaction $t$ at the conclusion of the game. This bribe is determined by a function $\mu_{BP}^{Bribe}$ with input $(s_t,b_t,p_t, H, M)$ that includes transaction's public characteristics, history of blocks $H$ and the mempool $M$. This function belongs to a set $\text{Bribe}^{BP}$. This set includes all the bribe functions of the block producer that we want to consider.  
\paragraph{Prior Knowledge}
All parties know their type and the set $\text{Bribe}^{BP}$.
\paragraph{Allocation rules}
The allocation rule is defined similarly to the Roughgarden's model, with the key distinction that a unique allocation rule exists for each potential block producer type.
\begin{definition} \label{allocation_produ} Allocation rule for the block producer of type $\mu_{BP}^{Bribe}$ is a vector-value function $x^{BP,\mu_{BP}^{Bribe}}$ that takes as input $H,M$ and  outputs $x^{BP,\mu_{BP}^{Bribe}}_t(H,M)\in \{0,1\}$ for every transaction $t \in M$. When the type of the block producer is implied by the context, we write $x^{BP}$ and $x^{BP}_t$ respectively. $x_t^{BP}(H,M)=1$ indicates that the block producer has included $t$ in their block, and $x_t^{BP}(H,M)=0$ the opposite.
\end{definition}
\paragraph{Payment and burning rules}
Let $B_k$ be the block that the block producer constructs.
The payment and burning rules are defined in accordance with the Roughgarden's model. Recall that the payment rule determines the compensation the block producer receives for each transaction included in their block from the total transaction fee (the bid), while the burning rule specifies the portion of that fee that is permanently removed from circulation (burnt).
\begin{definition}[Payment rule] The payment rule is a function $p^{BP}$ that takes as input $(H,B_k)$ and outputs the payment $p^{BP}_t(H,B_k)$  of the block producer per unit of size for every transaction $t \in B_k$. 
\end{definition}
\begin{definition}[Burning rule] The burning rule  is a function $q$ that takes as input $(H,B_k)$ and outputs the burnt amount per unit of size for every transaction $t \in B_k$. 
\end{definition}

\paragraph{Utilities}
The utility functions reflect the profit of the parties and are defined similarly to those in the Roughgarden's model, with two primary differences. The first is the inclusion of a "bribe loss" term: if a block producer fails to comply with an adversary's demand to censor a transaction, they forfeit a potential payment equal to the bribe they would have otherwise earned. This loss represents the opportunity cost of choosing honesty (or failing to censor) over the malicious bribe. The second difference is that users lack certain knowledge of the block producer's type and must instead operate based on their beliefs. Consequently, utility is defined as it is in the interim stage of a Bayesian game. For each user type, the utility is calculated as the expected value across all potential block producer types, weighted by the user's subjective probability that each specific type occurs (as specified by $\text{Belief}_{\text{User}_i \rightarrow BP}$).
\begin{definition}[Bribe Loss] For every transaction $t \in B_k \cap M$, the block producer of type $\mu_{BP}^{Bribe}$ incurs bribe loss $\mu_{BP}^{Bribe}(s_t,b_t,p_t, H, M)$. This means that the block producer will not receive the bribe for this transaction because they failed to comply with the briber's demand to censor it.
\end{definition}

\paragraph{Bribery-aware Transaction Fee Mechanism (Bribery-aware TFM)}
Given a set of candidate bribe functions $\text{Bribe}^{BP}$ and the Beliefs of the users regarding the types of the other users and the block producer (recall that for user $i$ these beliefs are denoted by $\text{Belief}_{\text{User}_i \rightarrow \text{Users}}$, $\text{Belief}_{\text{User}_i \rightarrow BP}$), we define a Bribery-aware Transaction Fee Mechanism as the tuple 

\begin{align*}
&(\{x^{BP,\mu_{BP}^{Bribe}}\}_{ \mu_{BP}^{Bribe}\in \text{Bribe}^{BP}}, p^{BP}, q)\}
\end{align*}

\paragraph{Properties of the Bribery-aware Transaction Fee Mechanism}
We define the properties that a Bribery-aware Transaction Fee Mechanism can satisfy. The first three properties are adaptations of the original Roughgarden's model. However, we extend these definitions to account for all potential block producer types. Finally, we introduce a fourth property which ensures that the mechanism is censorship resistant under an external malicious briber.
\begin{definition}[Dominant-Strategy Incentive Compatible (DSIC)]
A Bribery-aware Transaction Fee Mechanism is DSIC under the set of candidate bribe functions $\text{Bribe}^{BP}$ and the Beliefs of the users regarding the types of the other users and the block producer if the following holds:
Assuming the block producer of all types in $\text{Bribe}^{BP}$ follows the indicated allocation, every user has a dominant strategy no matter their beliefs for the transactions and the types of the other users.
\end{definition}

\begin{definition}[Myopic Block Producer Incentive Compatible (MBIC)]
A Bribery-aware Transaction Fee Mechanism is MBIC under the set of candidate bribe functions $\text{Bribe}^{BP}$ and the Beliefs of the users regarding the types of the other users and the block producer if the following holds:
For every $H,M$, for every type in $\text{Bribe}^{BP}$, the best response for the block producer of this type is to follow the indicated allocation rule and refrain from adding fake transactions to the mempool and to their block.
\end{definition}
\begin{definition}[Off-Chain Agreement-proof (OCA-proof)] A Bribery-aware Transaction Fee Mechanism is OCA-proof under the set of candidate bribe functions $\text{Bribe}^{BP}$ and the Beliefs of the users regarding the types of the other users and the block producer if the following holds: a coalition between any type of the block producer and the users cannot increase all their utilities via off-chain payments. 
\end{definition}

\begin{definition}[Censorship resistant]
    A Bribery-aware Transaction Fee Mechanism is Censorship resistant under the set of candidate bribe functions $\text{Bribe}^{BP}$ and the Beliefs of the users regarding the types of the other users and the block producer if the following holds: The allocation rules of all the types in $\text{Bribe}^{BP}$ of the block producer are bribe function agnostic (ignore the bribe functions).
\end{definition}
Note that if the allocation rule of the underlying TFM is bribe function agnostic, then the bribe functions under which the incentive properties hold reflect the bribery attacks  against which the TFM remains censorship resistant.

\section{Bribery-aware and Asymmetric Multiple Proposer Transaction Fee Mechanism Design}
\label{formalmodel}
In this section, we extend the bribery-aware TFM framework introduced in Section \ref{briberyaware} to accommodate an asymmetric multiple proposer design, thereby capturing the architectural requirements of FOCIL. Recall that in Section \ref{appendixcompariso}, we provide a high level overview describing how these extensions generalize the bribery-aware TFM to this setting.
\begin{itemize}

\item Let $H$ be the sequence of the preceding blocks $B_1, \ldots, B_{k-1}$. The new block is denoted by $B_k$.
\item Let $M$ be the mempool with all the available transactions the includers and the block producer can include in their inclusion lists and their block respectively. Following \cite{roughgarden2023transactionfeemechanismdesign}, we assume that all the parties have the same view on $M$. Let $M_0$ be the mempool that consists only of the transactions of the users (not the fake transactions originated from the includers and the block producer). Transactions in $M_0$ are considered candidates for censorship.
\item  Let $C_{Block}$ be the maximum size for $B_k$ and $C_{Incl}$ be the maximum size for an inclusion list.
\item Let $m$ be the number of includers and $n$ the number of the users.
\end{itemize}

\paragraph{Types of the user}
Every user is of a specific type determined by the characteristics of the transaction they want to send. Every transaction has the following characteristics: (i)  \textit{size} $s_t$; (ii) \textit{value} $v_t$ per unit of size; (iii) the bid $b_{t}$  per unit of size and (iv) \textit{public information} $p_{t}$. These characteristics are the same as those defined in Section \ref{briberyaware}. 
\paragraph{Beliefs of the user}
Every user $i$ has a belief about (i) which is the type of every other user denoted by $\text{Belief}_{\text{User}_i \rightarrow \text{Users}}$, and (ii) which is the type of every includer and the block producer denoted by $\text{Belief}_{\text{User}_i \rightarrow CM,BP}.$

\paragraph{Types of the includer}
 Each includer $j$ has a strategic type determined by the bribe they receive from the malicious briber for a transaction $t$ at the conclusion of the game. This bribe is contingent upon their decision to exclude transaction $t$, and is calculated based on their own strategy in coordination with the strategies of the includers and the block producer. Formally, this bribe is determined by a function $\mu_{CM_j}^{Bribe}$ with input $(s_t,b_t,p_t,\vec \alpha_t, H, M_0)$. This function belongs to a set $\text{Bribe}^{CM}$. $\text{Bribe}^{CM}$ includes all the bribe functions of the includers that we want to consider. The input $(s_t,b_t,p_t,\vec \alpha_t, H, M_0)$ is the same as those used in the bribe functions in Section \ref{briberyaware} with the difference that (i) part of the input is the original mempool $M_0$ that includes only the transactions submitted by the users and not those by includers and the block producer and (ii) part of the input are the strategies of all the includers and the block producer specified by $\vec\alpha_{t} \in \{0,1\}^{m+1}$ that is defined formally below. 
 \begin{definition}[Inclusion list vector]
When the types of the block producer and the includers are implied by the text, the inclusion list vector $\vec \alpha_{t} \in \{0,1\}^{m+1}$ indicates whether the block producer and includers included transaction $t$. If and only if the block producer included transaction $t$ in its block $B_{k}$, then $\vec \alpha_{t}[0] = 1$, and 0 otherwise. Similarly, if and only if the $j$th includer included transaction $t$ in its inclusion list, $\mathsf{IL}_{j}$, $\vec \alpha_{t}[j] = 1$, and 0 otherwise. Moreover, we define $\alpha:= \{\vec \alpha_t\}_{t \in B_k}$. When the types are not implied we write $\vec a_{t, \mu_{BP}^{Bribe}, \mu_{CM_1}^{Bribe}, \ldots, \mu_{CM_m}^{Bribe}} $.
\end{definition}
    
\paragraph{Beliefs of the includer}
Every includer has a belief about the type of the other includers and the block producer, denoted by $\text{Belief}_{CM}$. To simplify the notation, we assume that all the includers have the same beliefs. Otherwise, we can define the beliefs for every includer separately.

\paragraph{Types of the block producer}
The block producer has a type determined by the bribe they receive from the malicious briber for transaction $t$ at the end of the game. This bribe is determined by a function $\mu_{BP}^{Bribe}$ with input $(s_t,b_t,p_t,\vec \alpha_t, H, M_0)$. This function belongs to a set $\text{Bribe}^{BP}$. This set includes all the bribe functions of the block producer that we want to consider.  This function has the same inputs as the bribe function of the includers. 
\paragraph{Prior Knowledge}
All parties are aware of their own type and the sets $\text{Bribe}^{CM}, \text{Bribe}^{BP}$. The block producer knows the type of every includer. This assumption simplifies the model without weakening its implications, because the types of the includers affect the utility of the block producer only via the inclusion lists they create, which the block producer learns before constructing their block.

\subsection{Phases and Allocation Rules}
In this subsection, we outline the various phases of our model and describe how the allocation rules in Section \ref{briberyaware} are generalized to accommodate this assymetric multiple proposer setting. Moreover, we summarize the phases in Figure \ref{modelphases}.
\begin{figure*}[h]
    \centering
    \includegraphics[width=\textwidth]{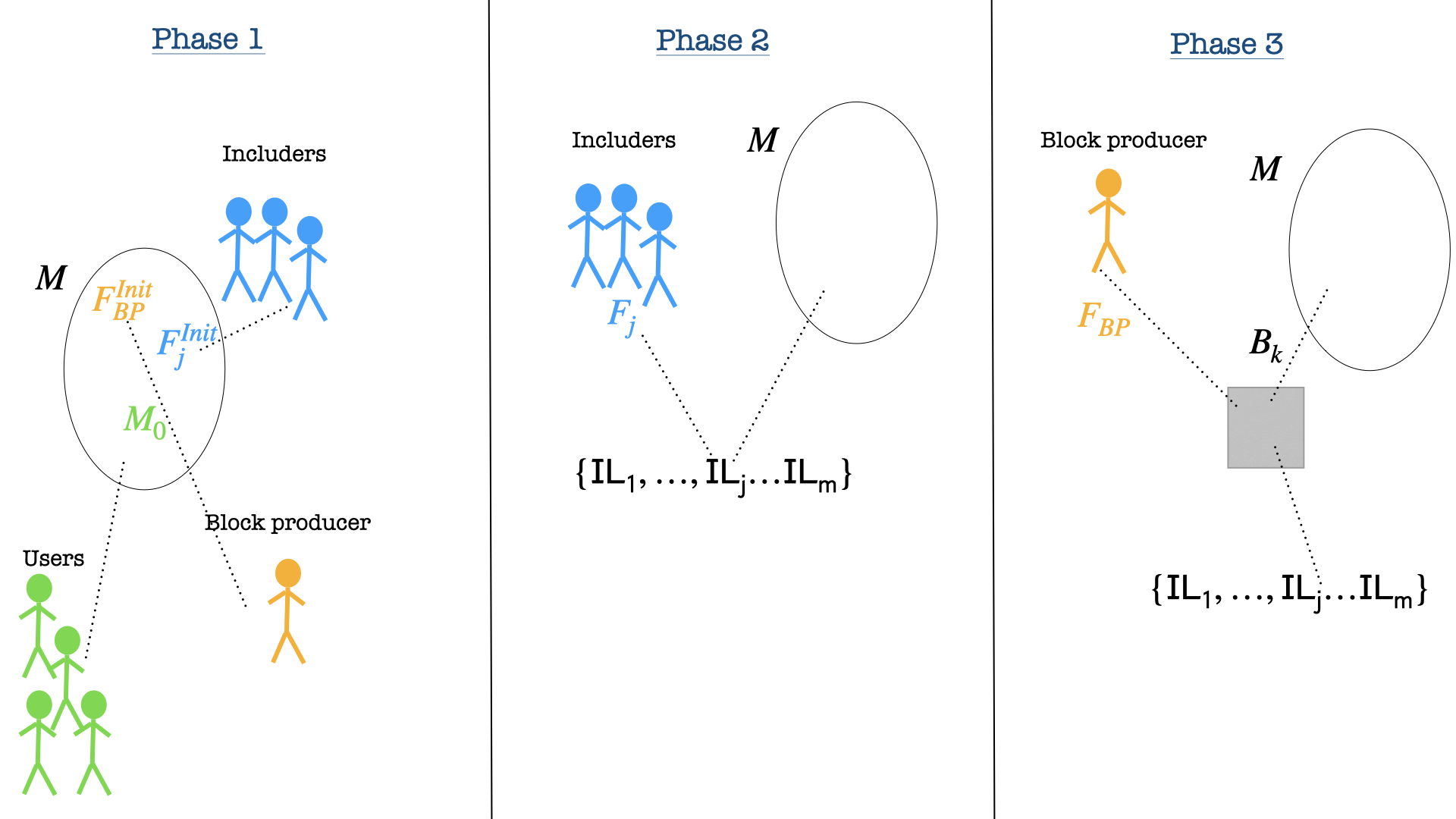} 
    \caption{Model Phases. In the first phase, users, includers, and the block producer send transactions to the mempool $M$. Transactions originating from the includers and the block producer are referred to as fake. In the second phase, includers create inclusion lists with transactions from the mempool M and/or fake transactions issued by themselves during this phase. In the image, we have considered that all includers have sent an inclusion list but this is not necessary. In the third phase, the block producer constructs their block including transactions from the mempool M, the inclusion lists and/or fake transactions issued by themselves during this phase. The sets $F^{Init}_{BP}, F^{Init}_j$ include the fake transactions submitted by the block producer and includer $j$ respectively during the first phase. $M_0$ is the set with transactions sent by the users. $F_j$ is the set with the fake transactions includer $j$ adds directly (without sending to the mempool) to their inclusion list and $F_{BP}$ the set with the fake transactions the block producer includes directly in their block.}
    \label{modelphases}
\end{figure*}

\paragraph{Phase $1$}
In the first phase, users, includers, and the block producer may send transactions to the mempool. Adopting the terminology of the Roughgarden's model, we classify all transactions submitted by the block producer or the includers in all phases as ``fake transactions" to distinguish them from those submitted by regular users. It is important to note that, despite this label, these transactions are legitimate and fully valid within the protocol. The strategies of the parties for this phase are as follows: 
\begin{itemize} 
 \item Strategy of user with transaction $t$: the bid $b_{t}$ per unit of size.
 \item Strategy of includer $j$ of type $\mu_{CM_j}^{Bribe}$: The set of fake transactions they send to the mempool, denoted by $F^{Init}_{j, \mu_{CM_{j}}^{Bribe}}$. 
 \item Strategy of the block producer of type $\mu_{BP}^{Bribe}$: The set of fake transactions they send to the mempool, denoted by $F^{Init}_{\mu_{BP}^{Bribe}}$. 
 \end{itemize}
 We use the abbreviation ``init" to distinguish fake transactions added to the mempool from those added directly in the inclusion lists or the final block respectively. Note that a unique strategy is defined for every potential type of each user, includer, and block producer. When the specific strategic types of the includers and the block producer are clear from the context, we simplify the notation for their initial fake transactions to $F^{Init}_{j}$ and $F^{Init}_{BP}$ respectively.
\paragraph{Phase $2$}
During this phase, every includer $j$ of type $\mu_{CM_j}^{Bribe}$ creates at most one inclusion list, denoted by $\mathsf{IL_{j,\mu_{CM_j}^{Bribe}}}$. When the type of the includer is denoted by the context we write $\mathsf{IL_{j}}$.  In this list, they can include transactions from $M$ and fake transactions created by themselves in this phase. The set of these fake transactions is denoted by $F_{j,\mu_{CM_j}^{Bribe}}$. When the type of the includer is implied by the context, the set is denoted by $F_{j}$. The includers send their inclusion lists to the block producer. The strategy of includer $j$ of type $\mu_{CM_j}^{Bribe}$ consists of (i) $F_{j,\mu_{CM_j}^{Bribe}}$ and (ii) the allocation rule defined in Definition \ref{allocation_incl} that determines which transactions from $M$ are included in $\mathsf{IL_{j,\mu_{CM_j}^{Bribe}}}$.
\begin{definition} 
\label{allocation_incl}
Allocation rule for an includer $j$ of type $\mu_{CM_j}^{Bribe}$ is a vector-value function $x^{j,\mu_{CM_j}^{Bribe}}$ that takes as input $(H,M)$ and outputs $x^{j,\mu_{CM_j}^{Bribe}}_t(H,M)\in \{0,1\}$ for every transaction $t \in M$. When the type of the includer $j$ is implied by the context, we write $x^j$ and $x^j_t$ respectively. $x_t^j(H,M)=1$ indicates that includer $j$ has included $t$ in their inclusion list, and $x_t^j(H,M)=0$ the opposite.
\end{definition}
\paragraph{Phase $3$}

In this phase, the block producer of type $\mu_{BP}^{Bribe}$ receives the list of inclusion lists $\mathsf{IL}=\{ \mathsf{IL_{1}}, \ldots, \mathsf{IL_{d}} \}$ where $d \leq m$ and creates a block $B_k$ with transactions from $M,\mathsf{IL}$ and/or fake transactions issued by themselves in this phase. The set of these fake transactions is denoted by $F_{\mu_{BP}^{Bribe}}$. If the type of the block producer is implied by the context we write $F_{BP}$.  The strategy of the block producer of type $\mu_{BP}^{Bribe}$ consists of (i) $F_{\mu_{BP}^{Bribe}}$ and (ii) the allocation rule defined in Definition \ref{allocation_prod} that determines which transactions from $M\cup  \mathsf{IL}$ will be included in $B_k$.

\begin{definition} \label{allocation_prod} Allocation rule for the block producer of type $\mu_{BP}^{Bribe}$ is a vector-value function $x^{BP,\mu_{BP}^{Bribe}}$ that takes as input $(H,M, \mathsf{IL})$ and inclusion lists $\mathsf{IL}$, and outputs $x^{BP,\mu_{BP}^{Bribe}}_t(H,M, \mathsf{IL})\in \{0,1\}$ for every transaction $t \in M \cup  \mathsf{IL}$. When the type of the block producer is implied by the context, we write $x^{BP}$ and $x^{BP}_t$ respectively. $x_t^{BP}(H,M, \mathsf{IL})=1$ indicates that the block producer has included $t$ in their block, and $x_t^{BP}(H,M,\mathsf{IL})=0$ the opposite.
\end{definition}
Following Roughgarden's approach we define\textit{ feasible} rules that respect maximum block size and maximum inclusion list size.
\begin{definition} 
\label{feasible}
\begin{itemize}
\item  An allocation rule $x^{j}$ is inclusion list-feasible if for every $H,M$ it holds:

$$
\sum_{t \in M}x_t^{j}(H,M)\cdot s_t \leq C_{Incl}
$$

\item  An allocation rule $x^{BP}_t$is block-feasible if for every $H,M,\mathsf{IL}$ it holds:

$$
\sum_{t \in M \cup \mathsf{IL} }x^{BP}_t(H,M,\mathsf{IL})\cdot s_t \leq C_{Block} 
$$

\item  A set of transactions $T$ is inclusion list-feasible if $\sum_{t \in T}s_t \leq C_{Incl}$.
\item  A set of transactions $T$ is block-feasible if $\sum_{t \in T}s_t \leq C_{Block}$.
\end{itemize}
\end{definition}
Note that the block producer may omit transactions from $\mathsf{IL}$  without having their block disregarded by attesters by strategically using fake transactions in the following ways: (i) the producer fills the block to its maximum capacity. Under the inclusion list rules, a full block is not disregarded by attesters even if it omits certain transactions from the inclusion lists; (ii) to execute a more elaborate deviation, the block producer spams the mempool with their own set of transactions $F^{Init}_{BP}$ so that includers select these instead of the transaction $t$ the producer wishes to censor. This allows the producer to satisfy the inclusion list requirements without including transaction $t$ in their block. Ordinarily, in this scenario, the block producer would be forced to include their own transactions that belong to $F^{Init}_{BP} \cap \mathsf{IL}$, and pay both the associated burn fees and the fees owed to the includers (see payment rule in Subsection \ref{payment}), otherwise their block would be disregarded by the attesters. To avoid this cost, the producer can include a specific ``invalidating" transaction early in the block (e.g., one that drains the sender’s balance). This renders the transactions in  $F^{Init}_{BP} \cap \mathsf{IL}$ invalid, allowing the producer to legally omit them from their block without violating inclusion list requirements.

\subsection{Payments and Costs}
\label{payment}
At the end of the game, the parties receive some payments and incur some costs based on their strategy and the strategies of the other parties.
\paragraph{Payment rules}

At a high level, the payment rules define the distribution of transaction fees for each transaction in the block, specifying the portions of the sender's payment awarded to the block producer and the includers, respectively. Note that the transaction sender pays the fees only if the transaction is included in a block that is not disregarded by the attesters. 
Formally, 
\begin{definition} \label{payment_producer} The payment rule for the block producer is a function $p^{BP}$ that takes as input $(H,B_k,\alpha)$ and outputs the block producer's payment $p^{BP}_t(H,B_k,\alpha)$ per unit of size for every transaction $t \in B_k$, provided the attesters approve the block $B_k$. 
\end{definition}
\begin{definition} \label{payment_includer} The payment rule for the includer $j$ is a function $p^{CM}$ that takes as input $(H, B_k, \alpha, j)$ and outputs the includer's $j$ payment $p^{CM}_t(H, B_k, \alpha, j)$ per unit of size for every transaction $t \in B_k$, provided the attesters approve the block $B_k$.
\end{definition}
\paragraph{Costs when the block is approved by the attesters}
Below we describe what costs the parties incur when the block is approved by the attesters. \par  
Firstly, the sender of the transaction (user) pays an amount that is burnt by the system (for instance, in Ethereum, this is the base fee, a protocol-computed reserve price determined by the burning rule of EIP-1559 \cite{roughgarden2020eip1559}). This amount is determined by the burning rule defined below in the same way as Roughgarden's model. 

\begin{definition}[Burning rule] The burning rule is a function $q$ that takes as input $(H,B_k)$ and outputs the burnt amount $q_t(H,B_k)$ per unit of size for every transaction $t \in B_k$. 
\end{definition}
Furthermore, we consider that the block producer may incur a fixed cost $\mu_{BP}^{Cost}$ per unit of size for every transaction they include in their block. This cost does not apply to their fake transactions. Depending on the underlying protocol, this cost may be assumed to be zero.
\paragraph{Costs regardless whether the block is approved by the attesters}
Assume that the types of the parties are implied by the context. Below, we describe what costs the parties incur regardless whether the block is approved by the attesters.
\begin{itemize}
\item For each fake transaction $t \in F^{Init}_{BP}$
  submitted to the mempool during Phase 1, the block producer incurs a cost equal to a fraction $\gamma$ of the total payment that would be awarded to the block producer and includers for that transaction ($p^{BP}_t(H,B_k,\alpha) \cdot s_t$ and $\sum_{j=1}^m p^{CM}_t(H,B_k,\alpha,j) \cdot s_t$ respectively). This captures the scenario where the block producer faces uncertainty regarding their selection as the creator of $B_k$, which would be necessary to successfully invalidate their initial fake transactions. As a result, they risk forfeiting the associated fees,  which would otherwise go to the creator of $B_k$ and the includers. If the block producer is certain they will be the creator of $B_k$, then $\gamma = 0$.
\item Consistent with the model we presented in Section \ref{briberyaware}, the includers and the block producer incur a bribe loss reflecting the external payment they forfeit when they fail to adhere to the briber's censorship requirements.
\begin{definition}[bribe loss] Assuming that the types of the parties are implied by the context:
\begin{itemize}
\item For every transaction $t \in B_k \cap M_0$, the block producer of type $\mu_{BP}^{Bribe}$ incurs bribe loss $\mu_{BP}^{Bribe}(s_t,b_t,p_t,\vec \alpha_t, H, M_0)$.
\item For every transaction $t \in \mathsf{IL_j} \cap M_0$ an includer $j$ of type $\mu_{CM_j}^{Bribe}$ incurs bribe loss $\mu_{CM_j}^{Bribe}(s_t,b_t,p_t,\vec \alpha_t,H,M_0)$.
\end{itemize}
\end{definition}
 \item Every includer incurs cost $\mu_{CM}^{Cost}$ per unit of size for every transaction they include in their inclusion list, not originated from them (regardless of whether this transaction is included in the block). Depending on the underlying protocol, this cost may be assumed to be zero.
\end{itemize}
\subsection{Utilities}
\paragraph{Utilities when the types of the parties are fixed and known}
The utilities of the parties represent their net profit, derived from the payment rules and cost defined in the preceding subsection.
In Definitions \ref{utility_us},\ref{utility_incl}, \ref{utility_pr} in Figures \ref{utility_user}, \ref{utility_producer}, we define formally the utilities of the parties when the types of all the parties are fixed (denoted by $\mu_{BP}^{Bribe},\mu_{CM_1}^{Bribe},\ldots \mu_{CM_m}^{Bribe}$) and known to all the other parties.

\begin{figure*}[h]

\begin{thickframe}
\begin{mdframed}
\underline{\textbf{Utility of the user when the types of the parties are fixed and known}:}\\[5pt]
At a high level, when the transaction of the user is included in the current block and the block is approved by the attesters, they gain the value of the transaction per unit of size. Moreover, they pay the transaction fee to the includers and the block producers (if any, depending on the payment mechanism), and the burning fee, per unit of size.

\begin{definition} \label{utility_us} Utility of a user who sent a transaction $t$ with value $v_t$ and size $s_t$.
\begin{itemize}
\item  If  $t \in B_k$  and $B_k$ meets
protocol’s predetermined criteria for inclusion lists (this means that it is approved by the attesters) then
    
    \begin{align*}
     & u_t(b_t, H, B_k,\alpha) = (v_t-p^{BP}_t(H,B_k, \alpha)-  \sum _{j\in \{1, \ldots, m\}}p^{CM}_t(H,B_k, \alpha,j)-q_t(H,B_k))\cdot s_t
    \end{align*}

\item  Otherwise: $u_t(b_t,H, B_k,\alpha)=0$.
\end{itemize}
\end{definition}
\vspace{5 mm}
\underline{\textbf{Utility of includer $j$ when the types of the parties are fixed and known}}: \\ [5pt]
At a high level, if $B_k$ is approved by the attesters, for every transaction in $B_k  $ not originated from the includer (this means that it does not belong to $F_{j} \cup F^{Init}_{j}$), the includer receives the corresponding transaction fee. In addition, for every transaction from $M_0$ in their inclusion list, they lose the bribe. Furthermore, for every transaction in their inclusion list, not created by them, they lose the per unit of size cost $\mu_{CM}^{Cost}$. 

   Finally, if $B_k$ is approved by the attesters, they pay the burning fee and the payment for the block producer and the other includers, per unit of size for every fake transaction in $B_k$ they have submitted to the mempool or they have included in their inclusion list.

\begin{definition} \label{utility_incl} Utility of the includer $j$.
\begin{itemize}
\item  If $B_k$ meets
protocol’s predetermined criteria for inclusion lists:
    \begin{align*}& u_{CM}(b_t,H,B_k,\alpha,j)=  \sum_{t \in B_k \cap  \neg ( F_{j} \cup F^{Init}_{j})}[p^{CM}_t(H,B_k,\alpha,j) \cdot s_t]\\& - \sum_{t \in\mathsf{IL_j} \cap M_0}[\mu_{CM_j}^{Bribe}(s_t, b_t,p_t,\vec \alpha_t, H, M_0)]-\sum_{t \in \mathsf{IL_j}\cap \neg ( F_{j} \cup F^{Init}_{j})}[\mu_{CM}^{Cost} \cdot s_t]\\&-\sum_{t\in B_k \cap (F_{j}\cup F^{Init}_{j}) }[q_t(H,B_k) +p^{BP}_t(H,B_k, \alpha) + \sum _{i\in \{1, \dots, m\}\setminus j }p^{CM}_t(H,B_k, \alpha,i)]\cdot s_t \end{align*}
\item Otherwise: The same as above with the difference that the first and the fourth term are equal to $0$. 
\end{itemize}

\end{definition}

\end{mdframed}
\end{thickframe}
\caption{Utility of the users and includers}
\label{utility_user}
\end{figure*}

\begin{figure*}[h]

\begin{thickframe}
\begin{mdframed}
\underline{\textbf{Utility of the block producer}}: 
At a high level, if $B_k$ is approved by the attesters, for every transaction in $B_k$, the block producer receives the corresponding fee and incurs the cost $\mu_{BP}^{Cost}$, per unit of size (except for the transactions that originated from the block producer). Moreover, regardless whether $B_k$ is approved by the attesters, for every transaction in $B_k \cap M_0$, the block producer loses the bribe. Additionally, if $B_k$ is approved by the attesters, for every fake transaction in $B_k$ created by them, they pay the burning fee and the fee for the includers per unit of size (recall that the burning fee and the fee for the includers are paid by the sender of the transaction; thus the block producer pays these amounts only for the transactions created by them - fake transactions).  Finally, for every transaction they submitted to the mempool in Phase $1$, they pay $\gamma$ fraction of the fee that corresponds to the block producer and the committee. 
 
\begin{definition}
\label{utility_pr}

Let $Fee_t$ be the total fee per unit of size that a transaction $t$ gives to the includers and the block producer when it is included in both an inclusion list and a block. The structure of this fee depends on the specifics of the bid $b_t$ and the fee mechanism under examination. 

\begin{itemize}
\item If $B_k$ meets
protocol’s predetermined criteria for inclusion lists:
\begin{align*}
& u_{BP}(b_t,H,B_k,\alpha)=\sum_{t \in B_k \cap \neg (F_{BP}\cup F^{Init}_{BP})}[(p^{BP}_t(H,B_k,\alpha)\\&  -  \mu_{BP}^{Cost}) \cdot s_t]  -\sum_{t \in B_k \cap M_0}[ \mu_{BP}^{Bribe}(s_t,b_t,p_t,\vec \alpha_t, H, M_0)] \\& -\sum_{t\in B_k \cap (F^{Init}_{BP}\cup F_{BP}) }[(q_t(H,B_k) +\sum _{i\in \{1, \dots, m\}}p^{CM}_t(H,B_k, \alpha,i))]\cdot s_t \\&  -\sum_{t \in F^{Init}_{BP}}[\gamma \cdot Fee_t \cdot s_t]  
\end{align*}
\item Otherwise: The same with the difference that the first and the third term are equal to $0$. 
\end{itemize}
\end{definition}

\end{mdframed}
\end{thickframe}
\caption{Utility of the block producer}
\label{utility_producer}
\end{figure*}

\paragraph{Utilities when the users and the includers are not aware of the type of the other includers and the block producer}
The utilities in this case are defined as in the interim stage of a Bayesian game. For each user and includer type, the
utility is calculated as the expected value across all potential block
producer and includer types weighted by the party’s
subjective probability that each specific type occurs (as specified
by their beliefs).

\subsection{Formal Definition of Bribery-aware and Asymmetric Multiple
Proposer Transaction Fee Mechanism and its Properties}
Assume 
\begin{itemize}

\item Sets of candidate bribe functions $\text{Bribe}^{CM},\text{Bribe}^{BP}$.
\item Beliefs of the parties about the bribe functions of the other parties, denoted by $$\{\text{Belief}_{\text{User}_i \rightarrow \text{Users}}, \text{Belief}_{\text{User}_i \rightarrow \text{CM,BP}}\}_{i \in \{1,...n\}}, \text{Belief}_{CM}$$

\item Parameter $\gamma$ related to the probability that the block producer is not the creator of the block.

\end{itemize}
\paragraph{Bribery-aware and Asymmetric Multiple Proposer Transaction Fee Mechanism}
A Bribery-aware and Asymmetric Multiple Proposer Transaction Fee Mechanism (TFM) under the above sets and parameters is the following tuple : 

\begin{align*}
&(\{x^{BP,\mu_{BP}^{Bribe}}\}_{ \mu_{BP}^{Bribe}\in \text{Bribe}^{BP}}, \\& \{x^{j, \mu_{CM_j}^{Bribe}}\}_{\forall \mu_{CM_j}^{Bribe}\in \text{Bribe}^{CM}},  p^{CM}, p^{BP}, q)
\end{align*}

Note that for simplicity, we consider TFMs where all the includers of the same type have the same allocation rule. 
We will examine Bribery-aware and Asymmetric Multiple Proposer Transaction Fee Mechanisms under the following properties:
\paragraph{Dominant-Strategy Incentive Compatible (DSIC)}
This definition is similar to the one presented in Section \ref{briberyaware} (which generalized \cite{roughgarden2021tfm} to multiple block producer types) but it incorporates includers. At a high level, it examines whether the users have a dominant strategy assuming that all the types of the includers and the block producer follow the indicated allocation rules.
\begin{definition}    
A Bribery-aware and Asymmetric Multiple Proposer TFM is DSIC under the above sets and parameters if the following holds: Assuming that the includers and the block producer of all types in $\text{Bribe}^{CM}, \text{Bribe}^{BP}$ follow the indicated allocation, every user has a dominant strategy no matter their beliefs for the transactions and the types of the other users.
\end{definition}

\paragraph{Myopic Committee Bayesian-Nash Incentive Compatible (MCBN)}
This definition examines whether an includer has incentives to follow the indicated allocation rule and refrain from adding fake transactions assuming that all the other includers and the block producer of all types do the same (regardless of the content of the mempool and the previous blocks). 

\begin{definition}

A Bribery-aware and Asymmetric Multiple Proposer TFM is MCBN under the above sets and parameters if the following holds for every includer $j$:

For every $H$, $M_0$, assume that: (i) every type of block producer in $\text{Bribe}^{BP}$ follows the indicated allocation rule and does not add any fake transactions to their block and the mempool and (ii) every type of the other includers in $\text{Bribe}^{CM}$ follows the indicated allocation rule and does not add any fake transactions to their inclusion list and the mempool. Then, for every type in $\text{Bribe}^{CM}$, the best response for the includer $j$ of this type, based on their beliefs $\text{Belief}_{CM}$, is to follow the indicated allocation rule and refrain from adding fake transactions to the mempool and their inclusion list (this means that $F_j, F^{Init}_{j}$ are empty).
\end{definition}

\paragraph{Myopic Block Producer Bayesian-Nash Incentive Compatible (MBBN)}
This definition examines whether a block producer has incentives to follow the indicated allocation rule and refrain from adding fake transactions assuming that all the includers of all types do the same (regardless of the content of the mempool and the previous blocks). 

\begin{definition}
A Bribery-aware and Asymmetric Multiple Proposer TFM is MBBN under the above sets and parameters if the following holds: For every $H,M_0$, if all the types of the includers follow the indicated allocation rules and do not add any fake transactions to their inclusion lists and the mempool then: for every type in $\text{Bribe}^{BP}$, the best response for the block producer of this type is to follow the indicated allocation rule and refrain from adding fake transactions to the mempool and to their block (this means that $F_{BP}, F^{Init}_{BP}$ are empty). 
\end{definition}

\paragraph{Myopic Block Producer Incentive Compatible (MBIC)}
This definition examines whether a block producer has incentives to follow the indicated allocation rule and refrain from adding fake transactions regardless of the includers' strategies, and the content of the mempool and the previous blocks. 
\begin{definition}
A Bribery-aware and Asymmetric Multiple Proposer TFM is MBIC under the above sets and parameters if the following holds: For every $H,M_0$ and strategy of the includers, for every type in $\text{Bribe}^{BP}$, the best response for the block producer of this type is to follow the indicated allocation rule and refrain from adding fake transactions to the mempool and to their block (this means that $F_{BP}, F^{Init}_{BP}$ are empty).
\end{definition}

\paragraph{MBBN vs MBIC}

Note that MBBN is a weaker property than MBIC because it makes a (Nash equilibrium related) assumption for the strategy of the includers. Moreover, if MCBN is combined with MBBN, then this means that the strategy profile where (i) all the types of the includers follow the indicated allocation and they do not add fake transactions to the mempool and their inclusion lists and (ii) all the types of the block producer follow the indicated allocation rule and they do not add fake transactions to the mempool and their block is a Bayesian Nash equilibrium.

\paragraph{Censorship resistant}
This definition examines whether the indicated allocation rules of all the types of the includers and the block producer are bribe function agnostic.
\begin{definition}
A Bribery-aware and Asymmetric Multiple Proposer TFM is Censorship resistant under the above sets and parameters if the following holds: The allocation rules of all the includers' types in $\text{Bribe}^{CM}$ and block producer's types in $\text{Bribe}^{BP}$ are bribe function agnostic (ignore the bribe functions). 
\end{definition}

\paragraph{Fair-under-congestion}
This definition ensures that the payment rule does not disincentivize the block producer from filling their block with transactions during periods of congestion (recall that a payment rule that allocates the entire transaction fee to includers does not satisfy this condition). 

\begin{definition}

A Bribery-aware and Asymmetric Multiple Proposer TFM is fair-under-congestion under the above sets and parameters if the following holds:  Assume arbitrary $H$, $M_0$ that is not block-feasible (Def. \ref{feasible}) and types of includers and block producer in $\text{Bribe}^{CM}$, $\text{Bribe}^{BP}$ respectively. Let: (i) $\{\mathsf{IL_1}, \ldots, \mathsf{IL_d}\}$ be the inclusion lists if all the includers follow the indicated allocation rule and do not add fake transactions to the mempool and their inclusion list; (ii) $B_k$ be the block if the block producer follows the indicated allocation rule and does not add any fake transactions to the mempool and the block. Then, for every $t \in M_0 \setminus B_k$ , there is a bidding strategy the user could follow for their transaction to be included in at least one inclusion list and the block, assuming that all the other transactions remain the same and the includers and block producer adhere to the indicated allocation rules and they do not add any fake transactions.
\end{definition}

\section{FOCIL TFMs}

\indent We use our model to identify both suitable and unsuitable Bribery-aware and Asymmetric Multiple Proposer TFMs for FOCIL \cite{thiery2024eip7805}. For simplicity, we write TFM instead of Bribery-aware and Asymmetric Multiple Proposer TFM. In more detail, we consider two types of TFMs that could be used for FOCIL \cite{thiery2024eip7805} and find one that is unsuitable. All these TFMs adhere to the current Ethereum's burning rule (EIP-1559). Under this burning rule, an amount from the transaction fee---referred to as the base fee---is burnt. The base fee is adjusted based on the number of transactions in the previous block to regulate supply and demand (cf. \cite{buterin2019eip1559}). \par  In the first TFM, denoted by \textit{Double TFM}, the user specifies two fees: one that is given to the includers, and another given to the block producer. In the second TFM, called \textit{Single TFM}, the user specifies one, total fee and the system (i.e., the blockchain protocol) determines how the fee is split between the includers and the block producer. We treat the fraction of the fee allocated to the includers as a parameter and analyze how varying this parameter impacts censorship resistance. If this parameter is set to allocate the entire fee to the block producer, the system effectively implements a TFM for the current version of FOCIL where the transaction bid is the same as Ethereum's bid and thus does not reimburse the includers. The same holds if we remove the fee that corresponds to includers in the Double TFM. Not reimbursing the includers assumes that includers are altruistic and do not require financial incentives to include transactions into their inclusion lists. Finally, in the third TFM, referred to as \textit{Single Prioritized TFM}, the user also sets a single fee and the system determines how the fee is split, with priority given to the includers. The latter is unsuitable for FOCIL and is used as an example to prove a negative result using our framework.

\begin{figure*}[h]
    \centering
    
    \begin{tabular}{c|c|c|c}
         & Double TFM & Single TFM & Single Prioritized TFM  \\
         \hline
         DSIC under assumptions & x & x & n/a \\
         MCBN & x & x & n/a \\
         MBBN & x & x & n/a \\
         Universal Censorship Resistance & x & - & n/a \\
         Simple User Experience & - & x & n/a \\
         Fair-under-congestion & x & x & - \\
    \end{tabular}
    \caption{Overview of results. Properties that are satisfied are denoted by x; those that are not by -; n/a means we do not provide proofs. The assumptions for DSIC are similar to the ones in \cite{roughgarden2021tfm} : (i) users cannot overbid, meaning they do not submit a bid higher than their transaction's value; and (ii) the base fee is not excessively low, meaning that the number of transactions with values higher than $r+\mu^{Cost}_{BP}$ is less than the maximum number of transactions in a block.}
    \label{tab:resultsoverview}
\end{figure*}

\paragraph{Definitions and assumptions}
In the FOCIL TFMs we introduce, we use the following terms: \par

\begin{definition}
\label{fee}
 \underline{Block producer fee of a transaction $t$} is the amount the sender of $t$ will pay to the block producer if they include the transaction in their block. \par 

 \underline{Committee fee of a transaction $t$} is the amount that the sender of $t$ will pay to includers (committee) if the transaction is included in an inclusion list \textbf{and} the block.

\end{definition}
In the following sections, we denote by $r$ the burning fee per unit of size. Moreover,  for simplicity, we assume that all the transactions have the same size denoted by $s$, but in Appendix \ref{variation}, we explain how our theorems are affected if we remove this assumption. In Appendix \ref{variation}, we also discuss some variations of the results we present below; for example, we discuss how our results are affected if there is at least one includer who is honest and thus never censors, or in the case of unconditional inclusion lists.   

\par 
\subsection{Double and Single TFM}
Recall that a Bribery-aware and
Asymmetric Multiple Proposer Transaction
Fee Mechanism consists of the payment rules for the includers and the block producer respectively, the burning rule, and the allocation rule for every type of the block producer and includer. Both Double and Single TFM, as already mentioned, follow Ethereum's current burning rule. Moreover, their allocation rules are bribe function agnostic which means that they are censorship resistant under the bribe functions that we have determined. Below, we describe how transaction bids would need to adjust to allow for an includer fee and the other components of Double and Single. Moreover, in this section, we describe the properties Double and Single TFM satisfy (according to our theorems and proofs) and we summarize these in Table \ref{tab:resultsoverview}. Formal definitions, theorems, and proofs are provided in Appendix \ref{double} and Appendix \ref{single}.  
\subsubsection{Components of Double TFM}
\paragraph{Transaction bid structure and the payment rule of Double TFM}
In Double TFM we define transaction's bid as $b_t=(\delta^{CM}_t,\delta^{BP}_t,c_t)$, where $\delta^{CM}_t$is the maximum fee per unit of size for the includers, $\delta^{BP}_t$ is the maximum fee per unit of size for the block producer, and $c_t$ is the maximum amount per unit of size the user is willing to pay for all the fees including burning fee. Given this transaction bid, the payment rule allocates transaction's fee as follows: \par (i) The block producer fee and the committee fee (Definition \ref{fee}) of a transaction $t$ of size $s$ are equal to $y=\max \{\min \{ \delta_t^{BP} \cdot s, c_t\cdot s -r \cdot s\},0\}$ and $\max \{\min \{ \delta_t^{CM}\cdot s, c_t \cdot s-r \cdot s- y\},0\}$ respectively;\par  (ii) The entire committee fee is awarded to the includer (if any) with the smallest order who includes the transaction in their inclusion list. \par Note that if we remove $\delta^{CM}_t$ from the transaction bid and we allocate a zero amount to the includers, then both the transaction bid structure and the payment rule become similar to Ethereum’s current transaction bid and payment rule.
\paragraph{Allocation rule of Double TFM}
 For the type in  $\text{Bribe}^{BP}$, \textit{the allocation rule for the block producer} is the following: They select transactions satisfying ($c_t \geq r +\mu^{Cost}_{BP}, \delta_t^{BP}\geq \mu^{Cost}_{BP}$ ), prioritizing those with the highest block producer fee until they make their block full or there are no other available transactions. 
\par For every type in $\text{Bribe}^{CM}$, \textit{the allocation rule of every includer} is the following: Every includer with order $j$ (the best order is $1$) chooses the following deterministic algorithm: First, the includer computes the set of transactions that a block producer, adhering to the specified allocation rule, would include in their block. The includer then removes from the mempool all transactions not belonging to this set. Additionally, the includer excludes any transactions offering a committee fee lower than $\mu^{Cost}_{CM}\cdot s$.  After that the includer:
(i) computes the subset of transactions from the mempool $M$ that maximise the utility of includer with order $1$, if they are included in their inclusion list. As we have assumed that every transaction has the same size, this corresponds to the set of transactions offering the highest committee fees; (ii) removes these transactions from the mempool $M$ and repeats the same procedure for the includers with order $2, \dots, j-1$, and (iii) selects the transactions that yield the highest committee fees for inclusion.
 \subsubsection{Components of Single TFM}
 \paragraph{Transaction bid structure and payment rule of Single TFM}
 In Single TFM, transaction's bid is $b_t=(c_t)$, where $c_t$ is the maximum amount per unit of size the user is willing to pay. The system splits the fee as follows: after subtracting the burning fee, the block producer receives an amount that is equal to the cost they incur for processing each transaction-namely $\mu^{Cost}_{BP}\cdot s$. If the remaining amount after burning is less than  $\mu^{Cost}_{BP}\cdot s$, the block producer receives the entire remaining amount. Any residual amount beyond this is then distributed as follows: $z$  fraction of this amount is allocated to the committee and $(1-z)$ fraction to the block producer. In Appendix \ref{app:zminimumbribe}, we examine how $z$ affects the minimum bribe needed for a briber to make a transaction $t_0$ to be omitted from a block. At a high level, the value of $z$ that maximizes this minimum bribe is increasing in $r$ and depends on $r,c_{t_0}$, which means that it does not apply to every target transaction. In Single TFM, the committee fee is shared in the same way as in Double TFM: it is allocated to the includer with the smallest order who includes this transaction in their inclusion list, if this transaction is included in the block. If no includer includes this transaction in their inclusion list, the user does not pay the committee fee. Formally, the block producer fee is equal to $\max \{[\min \{ \mu^{Cost}_{BP}\cdot s, c_t\cdot s -r \cdot s\}+\max \{(c_t\cdot s -r \cdot s-\mu^{Cost}_{BP}\cdot s),0\}\cdot (1-z)],0 \}$, and the committee fee  equal to $\max \{(c_t\cdot s -r \cdot s-\mu^{Cost}_{BP}\cdot s),0\}\cdot z$.
\paragraph{Allocation rule of Single TFM}
The includers allocation rule remains the same. The block producer's allocation rule remains the same, except that in this case, they only verify whether $c_t \geq r +\mu^{Cost}_{BP}$, not if $\delta_t^{BP}\geq \mu^{Cost}_{BP}$, as the bid does not contain a separate $\delta_t^{BP}$ field.

\paragraph{Bribe functions used to evaluate Double TFM and Single TFM and the censorship resistance they reflect} In Appendix \ref{bribefunction}, we define formally the bribe functions under which the properties we describe below and we summarize in Table \ref{tab:resultsoverview} are satisfied. Note that in our theorems, the bribe functions are complex because (i) they are specifically designed to capture the maximum values for which the stated properties hold. This means that if a higher bribe is offered through these bribe functions, the corresponding properties no longer hold; (ii) they account for corner cases that are very rare (for instance, when the block is filled with transactions with nearly zero fees, etc.). \par  Recall that if the allocation rule of the underlying TFM is bribe function agnostic, then the bribe functions under which the incentive properties hold reflect the bribery attacks  against which the TFM remains censorship resistant. In our case, the bribe functions we define, together with their tightness, reflect the following. \par In most cases, when there is no congestion, the minimum bribe an external briber needs to
censor a transaction $t_0$, is roughly the block producer fee plus block producer's cost to perform the following deviation: to add fake transactions to the mempool, making the includers prefer these transactions
over $t_0$, only to later invalidate them with another fake transaction in their block. If there
is a positive probability $\gamma$ that the block producer is not the block creator of the current slot
(thereby successfully invalidating the fake transactions they added to the mempool) then: the
higher the committee fee, the more costly the latter deviation is. This indicates that a payment rule that assigns a fraction of the transaction fee to includers can increase censorship resistance. This means that FOCIL design can benefit from these TFMs when includers are rational rather than altruistic. When there is congestion, the minimum bribe is approximately the difference between the block
producer fee of $t_0$ and the next highest fee of the available transactions.
\paragraph{Properties of Double and Single TFM}
\label{dsicassumptions}
Both Double and Single TFM are MCBN and MBBN. Recall that these properties are related to the incentives of the includers and the block producer to follow the indicated allocation rule and refrain from adding fake transactions when all the other includers and/or the block producer do the same. \par Moreover, both Double and Single TFM are DSIC under the following assumptions: (i) users cannot overbid, meaning they do not submit a bid higher than their transaction's value; and (ii) the base fee is not excessively low, meaning that the number of transactions with values higher than $r+\mu^{Cost}_{BP}$ is less than the maximum number of transactions in a block. Under similar assumptions, the authors in \cite{roughgarden2021tfm} proved that current Ethereum's TFM is DSIC referring to it as ``usually" DSIC. Furthermore, both  Double and Single TFM satisfy fair-under-congestion property. \par 
Note that the analysis of these TFMs reflects the censorship resistance and incentive properties of the existing FOCIL design. Recall that a Single TFM with its parameter set to allocate the entire fee to the block producer or a Double TFM without the fee that corresponds to the includers effectively reflects the current's FOCIL design that does not modify the existing Ethereum's TFM and does not reimburse includers.

\paragraph{Trade-offs between Double TFM and Single TFM}
We define ``Simple User Experience" as the property that the user does not need to set an extra fee compared to the current Ethereum TFM. We define ``Universal Censorship Resistance" as the property that it can offer the same level of censorship resistance for every transaction if the user selects an appropriate fee.
Double TFM satisfies ``Universal Censorship Resistance" and not ``Simple User Experience" and Single TFM the opposite. This holds because: (i) in Single TFM, the value of $z$ that offers better censorship resistance for a transaction $t_0$ depends on $r, c_{t_0}$, which means that it does not apply to every target transaction, compared to Double TFM where the user is free to determine the fraction of the total fee that is allocated to includers and the block producer; (ii) in Single TFM, the user does not need to set two types of fees (one for includers and another for the block producer), as happens in Double TFM.    
\subsection{Single Prioritized TFM}
\label{unfair}
In this section, we prove that a TFM with the following characteristics, called \textit{Single Prioritized TFM}, is not fair-under-congestion if $\mu^{Cost}_{BP}>0$:
\begin{itemize}
 \item  The users set a single fee and the system splits this fee between the block producer and the committee so that: (i) If the transaction is included in an inclusion list and the block, the entire fee ia allocated to the committee, and (ii) if the transaction is included in the block but is not included in any inclusion list, the entire fee is awarded to the block producer
 \item  The indicated allocation rule for the block producer is to select the transactions that yield the highest block producer fees.
 \end{itemize}
 The proof is as follows: \par When the block producer follows the indicated allocation rule, they include only transactions that do not belong to an inclusion list in their block. This happens because if they include a transaction that belongs to an inclusion list, they receive zero fees and incur a cost $\mu^{Cost}_{BP}\cdot s$. Thus, if the block producer follows the indicated allocation rule, no bidding strategy can make this transaction be included in both an inclusion list and the block.
\section{Conclusion}
Our work extends Transaction Fee Mechanism design to account for censorship resistance under bribery attacks, a fundamental property of blockchain protocols. Our framework can be utilised to study TFMs in any blockchain protocol that can be analyzed using Roughgarden’s model \cite{roughgarden2021tfm}, and our techniques apply to both single and multiple proposer blockchain protocols. Furthermore, we apply and extend our framework to study  Fork-Choice Enforced Inclusion Lists(FOCIL) \cite{thiery2024eip7805}, a critical evolution in the Ethereum roadmap that serves as the  censorship resistance and consensus flagship for the upcoming Hegotá hard fork. Although recent works such as Garimidi et al. [FC’25] have extended TFM design to multiple proposer settings, they do not aim to capture censorship under bribery attacks and are not compatible with the structure of FOCIL. We study FOCIL under several TFMs, and prove and quantify how TFMs that reimburse includers can increase its censorship resistance when includers are rational rather than altruistic. Moreover, we show that FOCIL maintains key incentive compatibility guarantees of the Ethereum TFM \cite{roughgarden2021tfm} across both the current and our proposed TFMs. We note that off-chain agreements between the involved parties are outside the scope of this paper. Extending our model to capture such agreements is left for future work.

\section*{Acknowledgments}
The authors would like to thank Barnabé Monnot and 
Luca Zanolini for their helpful comments and suggestions. The first author was supported by a grant from the Ethereum Foundation for this research.
\bibliographystyle{splncs04}
\bibliography{references}

\appendix

\section{Notation}
We summarize the notation used in the following sections in Figure \ref{notation}.
 \begin{figure*}[h]

\begin{thickframe}
\begin{mdframed}

\underline{\textit{Notation}} 
\begin{itemize}
\item Let $t_0$ be the target transaction that the external briber seeks to censor by offering bribes to the includers and/or block producer.
\item Let $c_{Incl}:= \lfloor C_{Incl}/s \rfloor$ be the maximum number of transactions an inclusion list can store.
\item  Let $c_{block}:=\lfloor C_{block}/s \rfloor$ be the maximum number of transactions a block can store.
Let $r:=q_t(H,B_k)=q_t(H)$ be the burning fee per unit of size for a transaction $t$ included in $B_k$ according to EIP-1559. Recall that $H$ is the history of blocks.
\end{itemize}

\begin{itemize}

\item Let $w$ be the number of transactions in $M_0$.
\item Given $M_0$, we define the following two ordered lists of transactions:
\begin{itemize}
    \item $L_{BP}$: This list consists of the transactions in $M_0$ that have block producer fee no smaller than $\mu^{Cost}_{BP} \cdot s$. This list is ordered by the block producer fee. The ordering is decreasing. The ties break according to a deterministic rule. The  block producer fee corresponding to position $j$ is denoted by $f_{j,BP}$. If there is no position $j$ in the list, we consider that $f_{j,BP}=0$.
    \item $L_{CM,c_{block}}$: This list consists of the transactions in $M_0$ that (i) belong to the first $c_{block}$ positions in $L_{BP}$ (or to $L_{BP}$ if $L_{BP}$  has fewer than $c_{block}$ positions) (ii) have committee fee no smaller than $\mu^{Cost}_{CM}\cdot s$. This list is ordered by the committee fee. The ordering is decreasing. The ties break according to a deterministic rule. The committee fee corresponding to position $j$ is denoted by $f_{j,CM}$.
    \end{itemize}
    We assume that $t_0$ belongs to the first $\min \{ c_{Incl} \cdot m, size_{L_{CM, c_{block}}} \}$ positions in $L_{CM,c_{block}}$, where $size_{L_{CM, c_{block}}}$ is the size of $L_{CM,c_{block}}$. We adopt this assumption because, under the allocation rule defined below, any transaction that fails to meet these requirements would not be included in either the inclusion list or the block, even in the absence of a bribe.
\item Let $size_{L_{BP}}$ be the size of $L_{BP}$.
\item  Let $sum_{max,c_{block}}$be the sum of the rewards (block producer fee minus $\mu^{Cost}_{BP}\cdot s$ per transaction) the block producer will receive if they include the first $\min \{ c_{block}, size_{L_{BP}} \}$ transactions of $L_{BP}$ in their block.

\item Let $o$ be the order of $t_0$ in $L_{CM,c_{block}}$ and $o_{BP}$ be the order of $t_0$ in $L_{BP}$.
\item Let $f_{CM}, f_{BP}$ be the committee and the block producer fee of $t_0$ respectively.
\end{itemize}
\end{mdframed}
\end{thickframe}
\caption{Notation}
\label{notation}
\end{figure*}
\section{Double TFM}
\label{double}

\subsubsection{Formal Definition of Double TFM}
\label{bribefunction}
Let us first define the parameters, bribe functions and beliefs for the Double TFM. We assume an arbitrary parameter $\gamma$ related to the probability that the block producer is not the creator of the block.
\paragraph{Sets of candidate bribe functions $\text{Bribe}^{CM}, \text{Bribe}^{BP}$} 
For the formal definitions cf. Figure \ref{set}, \ref{set_2}. At a high level, when $w > c_{block}$ (which means that there is congestion), the block producer can censor the target transaction $t_0$ without risking that the attesters will ignore their block, simply by replacing it with another transaction in $M_0$, potentially with a lower block producer fee. Thus, in this case, the block producer does not have incentives to censor $t_0$ if the bribe is less than the difference between the block producer fee offered by $t_0$ and that of the transaction that replaces it.
\par When $w < c_{block}$, the cost for the block producer to omit $t_0$---and thus the bribe required to induce deviation---is higher and determined by the following cases: The block producer can omit $t_0$ without filling the block with fake transactions but they will lose all block rewards as their block will be ignored by the attesters. Alternatively, the block producer can fill their block with fake transactions to avoid their block being ignored by attesters but they will pay the burning fee $r \cdot s$ for every fake transaction they add to the block. Finally, they can add fake transactions to the mempool to make the includers ignore the target transaction. In this case they need to perform and pay the following:
\begin{itemize}
    \item  Assign a committee fee to the fake transactions that is sufficiently high to render them more profitable for includers than the target transaction.
    \item  Assign a block producer fee to the fake transactions such that includers are led to believe they will be included in the block. Note that includers have no incentive to include transactions that are unlikely to be selected by the block producer, as they would receive no fees from such transactions.
    \item  Add a fake transaction to the block to invalidate the previously inserted fake transactions in the mempool, thereby avoiding payment of the committee fee and the block producer fee they had set. They still pay a $\gamma$ fraction of these fees to account for the possibility that the block producer may not be the proposer for the current slot, in which case the transactions could be included in the subsequent block. The number of fake transactions that must be added to the block in order to invalidate those previously placed in the mempool---and thus the incurred cost---depends on whether multiple inclusion list entries from the same sender are permitted. 
        
    \end{itemize}


\begin{figure*}[h]
\begin{thickframe}
\begin{mdframed}
\underline{\textit{Set $\text{Bribe}^{CM}$:}} 
The $\text{Bribe}^{CM}$consists of the following three functions.

\begin{itemize}
        \item $\mu^{Bribe}_{1,CM}$: This function gives to the includer $f_{CM}- \max \{f_{g,CM},\mu^{Cost}_{CM}\cdot s \}$, where $g=c_{Incl} \cdot \lceil \frac{o}{c_{Incl}} \rceil +1$  if they omit this transaction (regardless of the strategy of the other includers). 
        \textit{Comment: Note that $g$ reflects the next most profitable transaction that can replace the censored one.}
        \item  $\mu^{Bribe}_{2,CM}$: This function gives to the includer $f_{CM}- \max \{f_{g,CM},\mu^{Cost}_{CM}\cdot s \}$, where $g=c_{Incl} \cdot \lceil \frac{o}{c_{Incl}}\rceil +1$ if they omit this transaction and all the other includers and the block producer do the same.
        \item  $\mu^{Bribe}_{3,CM}$: This function gives an $X$ to the includer if they omit this transaction (regardless of the strategy of the other includers). $X$ can be any non negative real number. 
\end{itemize}
\end{mdframed}
\end{thickframe}
\caption{Set $\text{Bribe}^{CM}$ }
\label{set}
\end{figure*}

\begin{figure*}[h]
\begin{thickframe}
\begin{mdframed}
\underline{\textit{Set $\text{Bribe}^{BP}$}:} 
      $\text{Bribe}^{BP}$consists of the following bribe function $\mu^{Bribe}_{1,BP}$.
\begin{itemize}
        \item When $w>c_{block}$: This function gives to the block producer $f_{BP}-\max \{f_{c_{block}+1,BP}, \mu^{Cost}_{BP}\cdot s \}$ if they omit the transaction regardless of the strategy of the includers.  
        
      \end{itemize}

    \begin{itemize}

        \item When $w \leq c_{block}$ and at most one transaction per sender is allowed to be added to an inclusion list, the function $\mu^{Bribe}_{1,BP}$ gives to the block producer: 
                \begin{align*}
                & \min \{ f_{BP}-\mu^{Cost}_{BP}  \cdot s +r \cdot s \cdot  (c_{block}-w+1),  sum_{max,c_{block}},   f_{BP} -\mu^{Cost}_{BP} \cdot s \\&+(c_{block} - o_{BP}+1) \cdot  \gamma  \cdot f_{BP}  +\lceil\frac{(c_{block} - o_{BP}+1)}{m}\rceil \cdot  r \cdot s , \\&  f_{BP} -\mu^{Cost}_{BP} \cdot s+  \gamma \cdot   (\min \{m \cdot c_{Incl}, size_{L_{CM, c_{block}}}\} -o+1) \cdot \\& (f_{CM} + \max \{f_{c_{block}-(\min \{m \cdot c_{Incl}, size_{L_{CM, c_{block}}}\} -o+1)+1,BP},  \mu^{Cost}_{BP}   \cdot s \} ) \\& + \lceil \frac{(\min \{m \cdot c_{Incl}, size_{L_{CM, c_{block}}}\} -o+1) }{m}\rceil \cdot r \cdot s  \}
                \end{align*}
\textit{Comment: The first term reflects the scenario in which the block producer fills the block with fake transactions and pays the burning fee. The second term captures the case where the block producer omits $t_0$ despite having space in the block; as a result, the block is ignored by attesters, and the producer loses block rewards. The third term reflects the scenario in which the block producer adds fake transactions to the mempool to make includers prefer them over the target transaction, and later invalidates them via another transaction in their block.
Note that fake transactions originating from the same sender can be invalidated by a single transaction, whereas fake transactions from different senders require one transaction per sender to be invalidated. Thus, in the case where at most one transaction per sender is allowed in an inclusion list, the block producer can refrain from censorship even if it receives slightly higher bribes, compared to the case where multiple transactions per sender are allowed.}
     \item When $w \leq c_{block}$ and multiple transactions per sender are allowed to be added to an inclusion list, the function $\mu^{Bribe}_{1,BP}$ gives to the block producer:
    
    \begin{align*}
    &\min \{ f_{BP}-\mu^{Cost}_{BP}  \cdot s +r \cdot s \cdot  (c_{block}-w+1),  sum_{max,c_{block}},  f_{BP} -\mu^{Cost}_{BP} \cdot s \\&+(c_{block}- o_{BP}+1) \cdot  \gamma  \cdot  f_{BP}  +r \cdot s , \\& f_{BP} -\mu^{Cost}_{BP} \cdot s+ \gamma \cdot  (\min \{m \cdot c_{Incl}, size_{L_{CM, c_{block}}}\} -o+1) \cdot \\& (f_{CM} + \max \{f_{c_{block}-(\min \{m \cdot c_{Incl}, size_{L_{CM, c_{block}}}\} -o+1)+1,BP},\mu^{Cost}_{BP}  \cdot s \} ) + r \cdot s  \}
    \end{align*}
    \end{itemize}

\end{mdframed}
\end{thickframe}
\caption{ Set $\text{Bribe}^{BP}$}
\label{set_2}
\end{figure*}

\paragraph{Types of the includers and the block producer} The block producer is of type $\mu^{Bribe}_{1,BP}$. The includer with order $\lceil \frac{o}{c_{Incl}}\rceil$ is of type $\mu^{Bribe}_{1,CM}$ or of type $\mu^{Bribe}_{2,CM}$. The other includers can be of type $\mu^{Bribe}_{1,CM}$ , $\mu^{Bribe}_{2,CM}$  or $\mu^{Bribe}_{3,CM}$.

Note that if the allocation rule of this TFM are followed by all includers, the includer with order $\lceil \frac{o}{c_{Incl}}\rceil$ would be the one who would include $t_0$ in their inclusion list.
\paragraph{ Beliefs of the parties about the bribe functions of the other parties} The beliefs of the users for the other users are arbitrary. The users and the includers know that the type of the block producer is $\mu^{Bribe}_{1,BP}$ with probability $1$. Every user and every includer believes that the every other includer is of type $\mu^{Bribe}_{1,CM}$ , $\mu^{Bribe}_{2,CM}$  or $\mu^{Bribe}_{3,CM}$ with arbitrary probability.
\paragraph{Double TFM}

The \textit{Double TFM} under the above parameters, bribe functions and beliefs is defined as follows:

\begin{align*}
&(\{x^{BP,\mu^{Bribe}_{BP}}\}_{ \mu^{Bribe}_{BP}\in \text{Bribe}^{BP}}, \\&  \{x^{j, \mu^{Bribe}_{CM_j}}\}_{\mu^{Bribe}_{CM_j}\in \text{Bribe}^{CM}},  p^{CM} p^{BP}, q)
\end{align*}

where: 
\begin{itemize}

\item  $\mu^{Bribe}_{BP}$ is always equal to $\mu^{Bribe}_{1,BP}$ defined above.
 \item  For every includer $j$, $\mu^{Bribe}_{CM_j}$ is equal to:
 \begin{itemize}
    \item  $\mu^{Bribe}_{1,CM}$ or $\mu^{Bribe}_{2,CM}$ if the order of the includer is equal to $\lceil \frac{o}{c_{Incl}}\rceil$.
    \item $\mu^{Bribe}_{1,CM}$ or $\mu^{Bribe}_{2,CM}$  or $\mu^{Bribe}_{3,CM}$ otherwise.
\end{itemize}
\item  Burning rule $q$ is the Ethereum’s burning rule (EIP-1559).
\item  $p^{CM}$: For every transaction not included in the block, they receive zero fees. For every transaction $t$ with a bid $b_t=(\delta^{CM}_t,\delta^{BP}_t,c_t)$ added to the block, the includer is paid as follows: If this includer is the member with the smallest order who has added this transaction to their inclusion list, they receive the entire committee fee-namely $\max \{\min \{ \delta_t^{CM}\cdot s, c_t \cdot s-r \cdot s- y\},0\}$, where $y=\max \{\min \{ \delta_t^{BP} \cdot s, c_t\cdot s -r \cdot s\},0\}$. Otherwise, they receive zero fees. 
\item  $p^{BP}$:  For every transaction included in their block, they receive the block producer fee, regardless of whether this transaction has been included in an inclusion list. Recall that this fee is equal to $y=\max \{\min \{ \delta_t^{BP} \cdot s, c_t\cdot s -r \cdot s\},0\}$.

\item  For the type in  $\text{Bribe}^{BP}$, \textit{the allocation rule for the block producer} is the following: They select transactions satisfying

 ( $c_t \geq r +\mu^{Cost}_{BP}, \delta_t^{BP}\geq \mu^{Cost}_{BP}$ ), prioritizing those with the highest block producer fee until they make their block full or there are no other available transactions. 

\item  For every type in $\text{Bribe}^{CM}$, \textit{the allocation rule of every includer} is the following: \par

Every includer with order $j$ (the best order is $1$) chooses the following deterministic algorithm: First, the includer computes the set of transactions that a block producer, adhering to the specified allocation rule, would include in their block. The includer then removes from the mempool all transactions not belonging to this set. Additionally, the includer excludes any transactions offering a committee fee lower than $\mu^{Cost}_{CM}\cdot s$.  After that the includer:
(i) computes the subset of transactions from the mempool $M$ that maximise the utility of includer with order $1$, if they are included in their inclusion list. As we have assumed that every transaction has the same size, this corresponds to the set of transactions offering the highest committee fees; (ii) removes these transactions from the mempool $M$ and repeats the same procedure for the includers with order $2, \dots, j-1$, and (iii) selects the transactions that yield the highest committee fees for inclusion.

\end{itemize}
Note that under this allocation rule, only one includer includes each transaction. However, this is sufficient to ensure censorship resistance, as a block with empty space is disregarded by the attesters if it omits a transaction that appeared in even a single transaction list.
\subsubsection{Properties of Double TFM}
Recall that current Ethereum's TFM  (EIP-1559) is proved to be \textit{``usually'' DSIC} in \cite{roughgarden2021tfm}, 
which means that it is DSIC under the following assumptions: (i) the users cannot overbid meaning they do not submit a bid higher than their transaction's value and (ii) the burning fee per unit of size for this slot is not excessively low meaning that the number of transactions with values higher than the base fee plus block producer cost to include the transaction is less than the maximum number of transactions in a block. 

\begin{theorem} Assume a history of block $H$ and a set of transactions $T$ of the same size. If (i) the users cannot overbid meaning they do not submit a bid higher than their transaction's value and (ii) the burning fee per unit of size $r$ is not excessively low (meaning that the number of transactions with values higher than $r+\mu^{Cost}_{BP}$ is less than the maximum number of transactions in a block) then: assuming that all the types of includers and the block producer adhere to the indicated allocation rules, the following bidding strategy for a transaction $t$: $b_t=(\delta^{CM}_t,\delta^{BP}_t,c_t)=(0, \mu^{Cost}_{BP}, \min\{v_t, r+\mu^{Cost}_{BP} \})$ constitutes a dominant strategy for every user in Double TFM, irrespective of their beliefs.
\end{theorem}
\begin{proof} As all the types of includers and block producer adhere to the indicated allocation rule, they ignore bribes. Moreover, all the types of includers and the types of the block producer follow the same indicated allocation rule. Thus the beliefs of the user for the types of includers and the block producer do not affect their utility.

 As the base fee is not excessively low, transactions with values higher than $\mu^{Cost}_{BP} +r$ are fewer than $c_{block}$. Moreover, as the users do not overbid, transactions with $c_t>r+\mu^{Cost}_{BP}$ are also fewer than $c_{block}.$ This means that a block producer who follows the indicated allocation rule, regardless of their type, will include a transaction in their block \textit{if and only if }it holds  $(\delta_t^{BP} \geq \mu^{Cost}_{BP}$  and $c_t \geq \mu^{Cost}_{BP} +r$ ). Furthermore, as the includers follow the indicated rule, they include a transaction in their inclusion list \textit{only if }the block producer fee is higher than $\mu^{Cost}_{BP} \cdot s$. 

We have the following cases:
\begin{itemize}
    \item  $\min\{v_t, r+\mu^{Cost}_{BP} \}=r+\mu^{Cost}_{BP}$ : The current utility of a user who bids $b_t=(\delta^{CM}_t,\delta^{BP}_t,c_t)=(0, \mu^{Cost}_{BP} , \min\{v_t, r+\mu^{Cost}_{BP}\})$  in this case is $v_t-r-\mu^{Cost}_{BP} \geq 0$ regardless of their beliefs, because their transaction will be included in the block by the block producer irrespective of the bids of the other users and the type of the block producer. The utility of the user when their transaction is not included in the block is $0$. Thus, the user can increase their utility only if they can set a lower $c_t$ and make their transaction still included in the block. The burning fee is $r$ which means that regardless of their choice of $\delta^{CM}_t,\delta^{BP}_t$, if they decrease $c_t$, the block producer fee becomes lower than $\mu^{Cost}_{BP}\cdot s$. This means that this transaction will not be included in the block.
\item  $\min\{v_t, r+\mu^{Cost}_{BP}\}=v_t$. The current utility of the user in this case is zero because the transaction will not be included in the block regardless of the type of the block producer and the other bids. Note that the block producer fee of this transaction is lower than $\mu^{Cost}_{BP}\cdot s$.  The only way for the user’s transaction to be included in the block is if the user increases $c_t$. However, if the user increases $c_t$, then their utility will become negative if this transaction is included in the block.
\end{itemize}
As a result bid $b_t$ maximises the user’s utility regardless of the other bids and their beliefs for the includers and the block producer.
\end{proof}
\begin{theorem}Double TFM is Myopic Committee Bayesian-Nash Incentive Compatible (MCBN), under the bribe functions and beliefs specified above. \end{theorem}

\begin{proof} Assume arbitrary $H,M_0$ and that every type of all but one includers and block producer follow the indicated allocation rule and do not add fake transactions. We need to prove that every type of the remaining includer denoted by $j$ cannot increase their utility by deviating from the indicated allocation strategy or by adding fake transactions to their inclusion list or to the mempool. \par 

For all the types of includer $j$, it holds that their beliefs do not affect their utility because all the types of the other includers and the block producer follow the same indicated allocation rule. \par 

Now we prove that the includer cannot increase their utility by deviating from the indicated allocation rule.

 \paragraph{If the order of $j$ is not $\lceil \frac{o}{c_{Incl}}\rceil$}
\begin{itemize}
    \item  $t_0$ is not in their inclusion list when they follow the indicated allocation rule, and thus they do not incur any bribe loss. This means that the bribe function that determines their type does not affect their utility. Note that when all the includers follow the indicated allocation rule, $t_0$ is included in the inclusion list of the member with order $\min \{\lceil \frac{o}{c_{Incl}}\rceil, m\}$.  For example, if $o=10$ and $c_{Incl}=3$, then $t_0$ is included in the inclusion list of the member with order $4$.
    \item  All their current transactions offer them a committee fee no smaller than cost $\mu^{Cost}_{CM}  \cdot s$. Thus, their utility cannot increase by omitting them.
    \item  They cannot add a transaction already included by an includer with a higher (worse) order to their inclusion list, because they have no space. If they choose to replace one (or more) of their current transactions with such a transaction, their utility cannot increase; the members with a worse order include transactions with lower committee fees than their current transactions.
    \item If they choose to add one (or more) transactions already included by an includer with a smaller order, then their utility will decrease because they will take zero fees for this transaction, and they will incur cost $\mu^{Cost}_{CM} \cdot s$. 
    \item If they choose to add a transaction not included by any includer, then their utility will decrease because this transaction either has a very low block producer fee and thus will not be included in the block (which means that it will give them zero fees) or has a lower committee fee.
\end{itemize}
\paragraph{ If the order of $j$ is equal to $\lceil \frac{o}{c_{Incl}}\rceil$}
\begin{itemize}
    \item The includer $j$ incurs bribe loss equal to $f_{CM}- \max \{f_{g,CM},\mu^{Cost}_{CM}\cdot s \}$, where $g=c_{Incl} \cdot \lceil \frac{o}{c_{Incl}} \rceil +1$ , as they have included $t_0$ in their inclusion list.
    \item If they omit $t_0$, their utility will not increase because they will gain the bribe loss plus $\mu^{Cost}_{CM}\cdot s$ , but they will lose $f_{CM}$.
    \item  If they replace $t_0$ with a transaction included by an includer with a smaller order then their utility will decrease, as they will receive zero fees from this transaction.
    \item  If they replace $t_0$ with a transaction included by an includer with a higher order, or a transaction not included in an inclusion list, then the maximum fee they will gain from this transaction will be $\max \{f_{g,CM}, \mu^{Cost}_{CM}\cdot s \}$ and they will lose $f_{CM}$. $f_{CM}- \max \{f_{g,CM}, \mu^{Cost}_{CM}\cdot s\}$ is no smaller than the bribe loss that they would gain. Thus, their utility cannot increase.
    \item  Regarding the other transactions, the proof that their utility cannot increase by deviating from the indicated allocation rule is the same as the case when the order of $j$ is not $\lceil \frac{o}{c_{Incl}}\rceil$.
\end{itemize}

Regardless of their order, includer $j$ cannot increase their utility by adding fake transactions to their inclusion list, because these transactions will give them no extra reward. Moreover, they cannot increase their utility by adding fake transactions to the mempool, because:
\begin{itemize}
\item When this fake transaction does not affect which transactions the other includers  or the block producer include in their inclusion list and block respectively, it cannot affect their utility.
\item In order for the fake transaction to affect which transactions the other includers  or the block producer include in their inclusion list and block respectively, they need to give a committee fee or a block producer fee that will be paid by includer $j$.
\begin{itemize}
    \item  When this fake transaction has a committee fee that makes another includer with a smaller order prefer it over another transaction $t'$, then their utility can be affected if they ``steal'' $t'$. However, their utility cannot increase because (i) the maximum  includer $j$ can gain by this deviation is the committee fee of the omitted transaction $t'$ and (ii) includer $j$ needs to pay a committee fee that is no lower than the committee fee of $t'$, the block producer fee and the burning fee of this omitted transaction $t'$ . This means that the utility of the includer will decrease.
    \item  When this fake transaction has a block producer fee that makes another includer with a smaller order prefer it over another transaction $t''$: recall that the includers who follow the indicated allocation rule choose a transaction only if it belongs to the $c_{block}$ transactions with the highest block producer fee (which is also higher than $\mu^{Cost}_{BP}  \cdot s$). In this case, the other includer will prefer the fake transaction over $t''$only if, after the addition of the fake transaction, $t''$ does not belong to the best $c_{block}$ transactions in terms of block producer fee. This means, that even if includer $j$ steals $t''$, this transaction will give them no reward because it will not be added to the block by the block producer (the block producer will prefer the fake transaction as well).
    \item When this fake transaction has a block producer fee that just makes the block producer omit some other transaction in favour of this fake transaction, then the utility of includer $j$ cannot increase.
\end{itemize}
\end{itemize}
\end{proof}

\begin{theorem}
    Double TFM is Myopic Block Producer Bayesian-Nash Incentive Compatible (MBBN) under the bribe functions and beliefs specified above. 

\end{theorem}
\begin{proof} Assume that all the types of includers follow the indicated allocation rule and they do not add fake transactions. This means that transaction $t_0$ has been included in an inclusion list. Moreover, if the block producer follows the indicated allocation rule and does not add any fake transactions then $t_0$ will be included in their block. \par 

We prove that the block producer cannot increase their utility by deviating from the indicated allocation rule or by adding fake transactions to the mempool or their block. Note that the utility of the block producer does not depend on their beliefs for the type of includers because all the types follow the same indicated allocation rule. \par  
\begin{itemize}
\item  The block producer cannot increase their utility by replacing transactions different from $t_0$ with other transactions because their indicated allocation rule chooses the set of transactions that offer them the highest block producer fee.
\item  The block producer cannot increase their utility by adding more transactions, because there is no other space or available transactions with a block producer fee of at least  $\mu^{Cost}_{BP}\cdot s$.  Recall that for every transaction in their block, they incur cost $\mu^{Cost}_{BP}\cdot s.$
\item The block producer cannot increase their utility by omitting transactions, because all their current transactions have a block producer fee of at least $\mu^{Cost}_{BP}\cdot s.$
\item The block producer cannot increase their utility by omitting or replacing transaction $t_0$ regardless of whether there is congestion ($w>c_{block}$) or not. 
\end{itemize}
The proof for the last claim is the following:
\paragraph{When $w>c_{block}$}
    \begin{itemize}
        \item If the block producer omits $t_0$ without adding a new transaction, then they will lose $ sum_{max,c_{block}}$ and they will gain the bribe loss which is no higher. Thus, their utility will not increase.
        \item If the block producer replaces $t_0$ with another transaction:
        \begin{itemize}
            \item If  $size_{L_{BP}} \geq c_{block} +1$: The most profitable deviation the block producer can perform is to replace transaction $t_0$ with the transaction that has position $c_{block}+1$ in $L_{BP}$ .  If they do so, they will gain the bribe loss $f_{BP}- \max \{f_{c_{block}+1,BP},\mu^{Cost}_{BP}\cdot s\}= f_{BP}- f_{c_{block}+1, BP}$ but they will lose $f_{BP}-f_{c_{block}+1,BP}$ because $t_0$ has a higher block producer fee. Moreover, the block producer incurs the same cost $\mu^{Cost}_{BP}\cdot s$ for both $t_0$ and the transaction with position $c_{block}+1$ in $L_{BP}$ . This means that their utility will not increase.
            \item If $size_{L_{BP}} < c_{block} +1$: There are no other transactions with block producer fee at least $\mu^{Cost}_{BP}\cdot s$ to replace $t_0$. If the block producer replaces $t_0$  with a transaction that has lower block producer fee than $\mu^{Cost}_{BP}\cdot s$, they will gain the bribe loss $f_{BP}- \max \{f_{c_{block}+1,BP},\mu^{Cost}_{BP}\cdot s\}= f_{BP}-\mu^{Cost}_{BP}\cdot s$, but they will lose an amount of at least  $f_{BP}-\mu^{Cost}_{BP}\cdot s$ from the difference between the block producer fee of $t_0$ and of the newly added transaction. Thus, their utility will not increase.
        \end{itemize}
    \end{itemize}
    \paragraph{When $c_{block} \geq w$} Recall that $o$ is the order of $t_0$ in $L_{CM,c_{block}}$ and $o_{BP}$ the order of $t_0$ in $L_{BP}$. We examine the following two variants of the FOCIL protocol:
    \begin{itemize}
        \item \textbf{Multiple transactions per sender are allowed to be added to an inclusion list}: The utility of the block producer when they do not deviate is $sum_{max,c_{block}}$ minus the bribe loss for the transaction $t_0$. Let us examine what deviations the block producer can make:
    \begin{itemize}
            \item  They can omit transaction $t_0$ (losing their fee $f_{BP}$) and add $c_{block}-w+1$ fake transactions to their block so that they do not get penalised by the attesters (recall that $w$ is the number of the transactions in the mempool including $t_0$). For every fake transaction, they need to pay $r \cdot s$ for the burning fee. The amount they gain via this deviation is equal to the bribe loss plus $\mu^{Cost}_{BP}\cdot s$, and the amount they lose is equal to $f_{BP}+r \cdot s \cdot (c_{block}-w+1)$.  As the bribe loss is no higher than $f_{BP} -\mu^{Cost}_{BP} \cdot s+r\cdot s \cdot (c_{block}-w+1)$, their utility cannot increase.
            \item They can omit transaction $t_0$ (losing its fee $f_{BP}$) and add no fake transaction. This means that they will lose their entire block rewards. Their utility in this case becomes $0$, which is not higher than their current utility as the bribe loss is no higher than  $ sum_{max,c_{block}}$.
            \item  They can omit transaction $t_0$ (losing its fee $f_{BP}$) and avoid penalisation by the attesters in the following way: by adding fake transactions to the mempool so that includers do not include $t_0$ in their inclusion lists, and later invalidating them via a single fake transaction. To make the committee ignore $t_0$,  the block producer needs:
            \begin{itemize}
                \item (First deviation): To exclude $t_0$ from the first $c_{block}$ positions of $L_{BP}$. They can do this by adding $c_{block}- o_{BP}+1$ fake transactions with a block producer fee at least $f_{BP}$.
                \item (Second deviation): To create $\min \{m \cdot c_{Incl}, size_{{L_{CM,c_{block}}}}\} -o+1$ fake transactions that (i) have committee fee at least $f_{CM}$, (ii) belong to the first $c_{block}$ transactions that give the highest block producer fee which means that they have a block producer fee at least $\max \{f_{y,BP}, \mu^{Cost}_{BP} \cdot s\}$, where $y=c_{block}-(\min \{m \cdot c_{Incl}, size_{L_{CM,c_{block}}}\} -o+1)+1$.
            \end{itemize}

        Note that here we take the worst-case scenario where the fake transactions that offer the same fee as the real transactions are preferred by the includers and the block producer. \par
        
        When the block producer performs the above deviations, they gain the bribe loss plus $\mu^{Cost}_{BP}\cdot s$, but they lose $f_{BP}$ plus $\gamma$ fraction of the fees of the fake transactions plus the base fee $r \cdot s$ for the transaction that invalidates the fake transactions. The later amount is equal to $(c_{block}- o_{BP}+1) \cdot  \gamma  \cdot  f_{BP}  +r \cdot s$ for the first deviation and equal to  $\gamma \cdot  (\min \{m \cdot c_{Incl}, size_{{{L_{CM,c_{block}}}}}\} -o+1) \cdot(f_{CM} + \max \{f_{y,BP}, \mu^{Cost}_{BP} \cdot s \}) + r \cdot s$ for the second deviation. This means that if the bribe loss is at most
        
        \begin{align*}
        &\min \{ (c_{block}- o_{BP}+1) \cdot  \gamma  \cdot  f_{BP}  +r \cdot s, \\& \gamma \cdot  (\min \{m \cdot c_{Incl}, size_{L_{CM,c_{block}}}\} -o+1) \cdot(f_{CM} \\& + \max \{f_{y,BP}, \mu^{Cost}_{BP} \cdot s \}) + r \cdot s \} +f_{BP} -\mu^{Cost}_{BP}\cdot s 
        \end{align*}
        
         the block producer cannot increase their utility by deviating. 
        \end{itemize}
        \item  \textbf{Single sender per inclusion list}: The proofs are the same except for the last point where the block producer tries to exclude $t_0$ from the inclusion lists by adding fake transactions to the mempools and later invalidating them with other transactions in the block. In this case, as every inclusion list can include transactions from a single sender then the cost for the block producer to deviate is higher. This happens because the block producer will need to add $\lceil Y/m \rceil$ transactions to the block to invalidate $Y$ fake transactions in the mempool. At most $m$ fake transactions from the same sender can be included in the inclusion lists (one per inclusion list).
        
        \end{itemize}
    
\end{proof}

\begin{theorem} Double TFM is fair-under-congestion under the bribe functions and beliefs specified above. 
\end{theorem}
\begin{proof} Assume arbitrary $H$, mempool $M_0$ that is not block-feasible, lists $L_{CM,c_{block}}, L_{BP}$, and types of includers and block producer in $\text{Bribe}^{CM}, \text{Bribe}^{BP}$ respectively. As we have assumed that every transaction has the same size, the fact that $M_0$ is not block-feasible means that $w > c_{block}$. Let
\begin{itemize}
\item $\{\mathsf{IL_1}, \ldots, \mathsf{IL_d}\}$ be the inclusion lists if all the includers follow the indicated allocation rule and do not add fake transactions to the mempool and their inclusion list.
\item $B_k$ the block if the block producer follows the indicated allocation rule and does not add any fake transactions to the mempool and the block.
\end{itemize}
 We want to prove that for every $t \in M_0 \setminus B_k$ there is a bidding strategy the user could follow for their transaction to be included in at least one inclusion list and the block, assuming that: (i) all the other transactions remain the same, and (ii) the includers and block producer follow the indicated allocation rules and they do not add any fake transactions. \par 
Note that if all the types of includers follow the indicated allocation rules and do not add fake transactions, the inclusion lists consist of  the transactions that belong to the first $\min \{m \cdot c_{Incl}, size_{L_{CM,c_{block}}}\}$ positions in  $L_{CM,c_{block}}$. Recall that these transactions belong to $c_{block}$ transactions with the highest block producer fee and their block producer fee is at least $\mu^{Cost}_{BP}\cdot s$.  This means that if the block producer follows the indicated allocation rule and does not add any fake transactions, they include all the transactions from the inclusion lists. Thus, the fact that $t$ does not belong to the block means that one of the following holds:
\begin{itemize}
\item It does not belong to $L_{CM,c_{block}}$ because it has a block producer fee lower than $\max \{\mu^{Cost}_{BP}\cdot s, f_{c_{block}, BP}\}$.
\item  It does not belong to $L_{CM,c_{block}}$ because it has a committee fee lower than $\mu^{Cost}_{CM}\cdot s$.
\item It belongs to $L_{CM,c_{block}}$, but it has a committee fee lower than  $f_{m \cdot c_{Incl},CM}$.
\end{itemize}
If the sender of transaction $t$  gives a bid  $b_t=(\delta^{CM}_t,\delta^{BP}_t,c_t)$, such that:
\begin{itemize}
\item  $c_t= \delta_t^{BP}+ \delta_t^{CM}+r$
\item  $\delta_t^{BP} \cdot s > \max \{ \mu^{Cost}_{BP}\cdot s, f_{c_{block}, BP} \}$
\item $\delta_t^{CM} \cdot s > \max \{f_{m \cdot c_{Incl},CM}, \mu^{Cost}_{CM} \cdot s\}$
\end{itemize}

their transaction will be included in both an inclusion list and the block. 
\end{proof}
Finally, the Double TFM is \textit{censorship resistant} because the allocation rules of all the types of the includers and the block producer ignore bribes. 
We discuss several variations of the above theorems. 
\subsection {Variation of our results}
\label{variation}
\paragraph{ Dominant strategy of the users if the inclusion lists were unconditional}
Note that in the current model, the inclusion lists are conditional. If they were unconditional and also the following hold: \begin{itemize}

\item  $c_{block} > m \cdot c_{Incl} \geq  w$
\item  The  indicated allocation rule for the block producer was to include all the transactions from the inclusion lists even if the transactions offer a block producer fee lower than $\mu^{Cost}_{BP}\cdot s$,
\end{itemize}
then under the same assumptions as the above theorem, the dominant strategy of the users would be:
\begin{itemize}
\item  $b_t=(\delta^{CM}_t,\delta^{BP}_t,c_t)=(0, \mu^{Cost}_{BP}, \min\{v_t, r+\mu^{Cost}_{BP} \})$, if $\mu^{Cost}_{BP} \leq \mu^{Cost}_{CM}$, and
\item $b_t=(\delta^{CM}_t,\delta^{BP}_t,c_t)=(\mu^{Cost}_{CM}, 0, \min\{v_t, r+\mu^{Cost}_{CM} \})$, if $\mu^{Cost}_{BP} > \mu^{Cost}_{CM}$.
\end{itemize}
The intuition behind this result is that, given the allocation rules followed by the includers and the block producer, any transaction included in an inclusion list will also be included in the block. Therefore, the user's most profitable strategy is to cover either the minimum cost required by the includer or the minimum cost required by the block producer.
\paragraph{Removing the assumption that all the transactions have the same size}

In our proofs, for simplicity, we assume that all the transactions have the same size. If we removed this assumption: 
\begin{itemize}

\item  Instead of selecting the transactions with the highest fee, the allocation rule of includers and the block producer would select the set of transactions that maximise the sum of the fees (taking into account also the size of the transaction sizes and $C_{block}, C_{Incl}$).
\item  For some mempools, includers would have incentives to add fake transactions to the mempool to ``steal'' a transaction $t_1$ from another member with a smaller order. To steal this transaction, they should create a fake transaction $t_f$ that has a size and a committee fee that make $t_f$ more desirable than $t_1$ for the other includer. This deviation is not significant in the context of censorship, as it does not negatively affect the utility of other includers; in fact, it results in a utility gain for both parties involved.
\end{itemize}
\paragraph{Allocation rules that do not exclude transactions with a block producer fee lower than $\mu^{Cost}_{BP}>0$}

Note that the above indicated allocation rules of both the includers and the block producer exclude all the transactions with a block producer fee lower than $\mu^{Cost}_{BP}>0$. \par 

However, when the block is not full, the block producer will lose the block rewards if they omit a transaction from the inclusion list, even if this transaction has a block producer fee lower than $\mu^{Cost}_{BP}>0$. This means that the indicated allocation rule of the block producer that excludes transactions with a block producer fee lower than $\mu^{Cost}_{BP}>0$ is not a dominant strategy. \par 

Moreover, at a Nash equilibrium, we can have a variation of the above indicated allocation rules where both the includers and the block producers include transactions with a block producer fee lower than $\mu^{Cost}_{BP}>0$ in the following case: if the block rewards are higher than the loss the block producer incurs by adding these transactions. Intuitively, this happens because it is more profitable for the block producer to incur some loss from this type of transaction to avoid rejection by the attesters.

\paragraph{The amount of bribe in the bribe functions is tight}

The amount of bribe in the bribe functions of the includers and the block producer we have defined is the maximum that can be set so that the theorems still hold.

\paragraph{The impact of adding some includers that always include the target transaction}

Even if we added this type of includer, the bribe of the block producer specified in their bribe function could not increase. The reason is that in this notion we examine whether the utility of the block producer increases when they deviate assuming that all the types of includers follow the indicated allocation rule, which already ignores bribes. \par 

However, adding these includers could decrease the bribe of the block producer when we try to prove that the Double TFM is MBIC. This notion examines the utility of the block producer for every strategy of includers. If there is no includer that includes the target transaction, then the maximum cost of the block producer to deviate is $f_{BP}-\mu^{Cost}_{BP}\cdot s$. If at least one member includes the target transaction in their block, their cost is the same as this one in the bribe function we use for proving MBBN.





\section{Single TFM}
\label{single}
\subsubsection{Notation and Formal Definition of Single TFM}
We adopt the same notation as in Section \ref{double}. In this setting, there is a single fee per transaction that is split according to a fixed rule across all transactions. As a result, the lists $L_{BP}, L_{CM,c_{block}}$ are both ordered by this unified fee. Consequently, ordering by the block producer fee or the committee fee yields the same result. Therefore, the order of $t_0$ in $L_{CM,c_{block}}$ (denoted by $o$)  and the order of $t_0$ in $L_{BP}$ (denoted by $o_{BP}$) are equal.  
The definition of Single TFM is the same as the definition of the Double TFM apart from the following:
\begin{itemize}
\item The definition of the payment functions. In this case, the payment functions are as follows:
\begin{itemize}
    \item  $p^{CM}$: For every transaction not included in the block they receive zero fees. For every transaction $t$ with bid $b_t=(c_t)$ added to the block, an includer is paid as follows: If this includer is the committe member with the smallest order who has added this transaction to their inclusion list, they receive the entire committee fee-namely $\max \{(c_t\cdot s -r \cdot s-\mu^{Cost}_{BP}\cdot s),0\}\cdot z$. Otherwise, they receive zero fees.
   
    \item  $p^{BP}$: For every transaction included in their block, they receive the block producer fee-namely $\max \{\min \{ \mu^{Cost}_{BP} \cdot s, c_t\cdot s -r \cdot s\}+\max \{(c_t\cdot s -r \cdot s-\mu^{Cost}_{BP} \cdot s),0\}\cdot (1-z) ,0\}$.
\end{itemize}
\item The block producer's allocation rule remains the same, except that in this case, they only verify whether $c_t \geq r +\mu^{Cost}_{BP}$, not if $\delta_t^{BP}\geq \mu^{Cost}_{BP}$, as the bid does not contain a separate $\delta_t^{BP}$ field.
\item  In this case, the relation between $f_{CM}, f_{BP}$ is determined by the system rather than the user.
\end{itemize}
\subsubsection{Properties of Single TFM}
\begin{theorem} Assume a history of block $H$ and a set of transactions $T$ that have the same size. If (i) the users cannot overbid meaning they do not submit a bid higher than their transaction's value and (ii) the burning fee per unit of size $r$ is not excessively low (meaning that the number of transactions with values higher than $r+\mu^{Cost}_{BP}$ is less than the maximum number of transactions in a block) then: assuming that all the types of includers and the block producer follow the indicated allocation rules, the following bidding strategy for a transaction $t$: $b_t=(c_t)=( \min\{v_t, r+\mu^{Cost}_{BP} \})$ is a dominant strategy for every user in Single TFM regardless of their beliefs.
\end{theorem}
\begin{proof} The proof is similar to the proof for the Double TFM. Note that the block producer fee and the committee fee under bidding strategy  $\min\{v_t, r+\mu^{Cost}_{BP} \}$  in the Single TFM and bidding strategy $(0, \mu^{Cost}_{BP}, \min\{v_t, r+\mu^{Cost}_{BP} \})$ in the Double TFM offer the same block producer and committee fee. In both cases, if the block producer include this transaction in their block, they will collect $\max \{(\min\{v_t, r+\mu^{Cost}_{BP} \}-r) \cdot s, 0\}$, and the includer will receive zero fees.
\end{proof}

\begin{theorem} Single TFM is Myopic Committee Bayesian-Nash Incentive Compatible (MCBN) under the bribe functions and beliefs specified above. \end{theorem}
\begin{proof}The first part of the proof, given below, is the same as the proof for the Double TFM because the committee fee is shared among the includers in the same way as in the Double TFM, $L_{BP}, L_{CM,c_{block}}$ consist of the same transactions, and the bribe functions are defined with respect to $f_{CM}, f_{BP}$. \par 

`` Assume arbitrary $H,M_0$ and that every type of all but one includers and block producer follow the indicated allocation rule and do not add fake transactions. We need to prove that every type of the remaining includer denoted by $j$ cannot increase their utility by deviating from the indicated allocation strategy or by adding fake transactions to their inclusion list or to the mempool. \par 

For both types of includer $j$, it holds that their beliefs do not affect their utility because all the types of the other includers and the block producer follow the same indicated allocation rule. \par 

\paragraph{If the order of $j$ is not $\lceil \frac{o}{c_{Incl}}\rceil$}
\begin{itemize}
    \item  $t_0$ is not in their inclusion list when they follow the indicated allocation rule, and thus they do not incur any bribe loss. This means that the bribe function that determines their type does not affect their utility. Note that when all the includers follow the indicated allocation rule, $t_0$ is included in the inclusion list of the member with order $\min \{\lceil \frac{o}{c_{Incl}}\rceil, m\}$.  For example, if $o=10$ and $c_{Incl}=3$, then $t_0$ is included in the inclusion list of the member with order $4$.
    \item  All their current transactions offer them a committee fee no smaller than cost $\mu^{Cost}_{CM}  \cdot s$. Thus, their utility cannot increase by omitting them.
    \item  They cannot add a transaction already included by an includer with a higher order to their inclusion list, because they have no space (if they had space, then the includers with a higher order would have no transactions in their inclusion lists). If they choose to replace one (or more) of their current transactions with such a transaction, their utility cannot increase; the members with a worse (higher) order include transactions with lower committee fees than their current transactions.
    \item If they choose to add one (or more) transactions already included by an includer with a smaller order, then their utility will decrease because they will take zero fees for this transaction, and they will incur cost $\mu^{Cost}_{CM} \cdot s$. 
    \item If they choose to add a transaction not included by any includer, then their utility will decrease because this transaction either has a very low block producer fee and thus will not be included in the block (which means that it will give them zero fees) or has a committee fee lower than $\mu^{Cost}_{CM}$.
\end{itemize}
\paragraph{ If the order of $j$ is equal to $\lceil \frac{o}{c_{Incl}}\rceil$}
\begin{itemize}
    \item The includer $j$ incurs bribe loss equal to $f_{CM}- \max \{f_{g,CM},\mu^{Cost}_{CM}\cdot s \}$, where $g=c_{Incl} \cdot \lceil \frac{o}{c_{Incl}} \rceil +1$ , as they have included $t_0$ in their inclusion list.
    \item If they omit $t_0$, their utility will not increase because they will gain the bribe loss plus $\mu^{Cost}_{CM}\cdot s$ , but they will lose $f_{CM}$.
    \item  If they replace $t_0$ with a transaction included by an includer with a smaller order then their utility will decrease, as they will receive zero fees from this transaction.
    \item  If they replace $t_0$ with a transaction included by an includer with a higher order, or a transaction not included in an inclusion list, then the maximum fee they will gain from this transaction will be $\max \{f_{g,CM}, \mu^{Cost}_{CM}\cdot s \}$ and they will lose $f_{CM}$. $f_{CM}- \max \{f_{g,CM}, \mu^{Cost}_{CM}\cdot s\}$ is no smaller than the bribe loss. Thus, their utility cannot increase.
    \item  Regarding the other transactions, the proof that their utility cannot increase by deviating from the indicated allocation rule is the same as the case when the order of $j$ is not $\lceil \frac{o}{c_{Incl}}\rceil$.
\end{itemize}
Regardless of their order, includer $j$ cannot increase their utility by adding fake transactions to their inclusion list, because these transactions will give them no extra reward. Moreover, they cannot increase their utility by adding fake transactions to the mempool, because:
\begin{itemize}
\item  When this fake transaction does not affect which transactions the other includers  or the block producer include in their inclusion list and block respectively, it cannot affect their utility.
\item  In order for the fake transaction to affect which transactions the other includers  or the block producer include in their inclusion list and block respectively, they need to have a fee that will be paid by includer $j$."
\par 
\textit{The remaining part of the proof has some differences and is as follows:}
\begin{itemize}
    \item  When this fake transaction gives a committee fee that makes another includer with a smaller order prefer it over another transaction $t'$, then their utility can be affected if they ``steal'' $t'$. However, their utility cannot increase because (i) the maximum includer $j$ can gain by this deviation is the committee fee of the omitted transaction $t'$ and (ii) includer $j$  needs to pay the committee fee of this omitted transaction $t'$, the burning fee, $\mu^{Cost}_{BP}$, and $(1-z)/z$ times the committee fee of $t'$ (this is the fee that is awarded to the block producer apart from $\mu^{Cost}_{BP}$). This means that the utility of includer $j$  will decrease.
    \item  When this fake transaction has a block producer fee that makes another includer with a smaller order prefer it over another transaction $t'$ (because $t'$, after the addition of the fake transaction, will not belong to the first $c_{block }$ positions of $L_{BP}$): In this case, also the committee fee of the fake transaction will be higher than the committee fee of $t'$. Thus, the utility of includer $j$ will decrease.
    \item  When this fake transaction has a block producer fee that makes the block producer omit some other transactions in favour of this fake transaction, then the utility of includer $j$ cannot increase.
\end{itemize}
\end{itemize}
\end{proof}
\begin{theorem} Single TFM is Myopic Block Producer Bayesian-Nash IncentiveCompatible (MBBN) under the bribe functions and beliefs specified above. \end{theorem}
\begin{proof} The proof is similar to the proof for the Double TFM, because the committee fee is shared among the includers in the same way as in the Double TFM, $L_{BP}, L_{CM,c_{block}}$ consist of the same transactions, and the bribe functions are defined with respect to $f_{CM}, f_{BP}$. The difference is that there is a dependency between $f_{CM}$ and $f_{BP}$ but this does not affect the proofs.   \end{proof}
\begin{theorem} Single TFM is fair-under-congestion under the bribe functions and beliefs specified above. \end{theorem}
\begin{proof}Assume arbitrary $H$, mempool $M_0$ that is not block-feasible, lists $L_{CM,c_{block}}, L_{BP}$ and types of includers and block producer in $\text{Bribe}^{CM}, \text{Bribe}^{BP}$ respectively. As we have assumed that every transaction has the same size, then the fact that $M_0$ is not block-feasible means that $w > c_{block}$. Let
\begin{itemize}
\item  $\{\mathsf{IL_1}, \ldots, \mathsf{IL_d}\}$ be the inclusion lists if all the includers follow the indicated allocation rule and do not add fake transactions to the mempool and their inclusion list.
\item $B_k$ the block if the block producer follows the indicated allocation rule and does not add any fake transactions to the mempool and the block.
\end{itemize}
 We want to prove that for every $t \in M_0 \setminus B_k$ there is a bidding strategy the user could follow for their transaction to be included in at least one inclusion list and the block, assuming that all the other transactions remain the same, and the includers and block producer follow the indicated allocation rules and they do not add any fake transactions. \par 

Recall that all the types of includers who follow the indicated allocation rules and do not add fake transactions include in their inclusion lists the transactions that belong to the first $\min \{m \cdot c_{Incl}, size_{L_{CM,c_{block}}}\}$ positions in  $L_{CM,c_{block}}$. This means that if the block producer follows again the indicated allocation rule and does not add any fake transactions, they include all the transactions from the inclusion lists. Thus, the fact that $t$ does not belong to the block means that at least one of the followings hold:
\begin{itemize}

\item  It does not belong to $L_{CM,c_{block}}$ because it has a block producer fee lower than $\max \{\mu^{Cost}_{BP}\cdot s, f_{c_{block}, BP}\}$.
\item  It does not belong to $L_{CM,c_{block}}$ because it has a committee fee lower than $\mu^{Cost}_{CM}\cdot s$.
\item It belongs to $L_{CM,c_{block}}$, but it has a committee fee lower than  $f_{m \cdot c_{Incl},CM}$.
\end{itemize}
If the sender of transaction $t$  gives a bid  $b'_t=(c'_t)$, such that:
\begin{itemize}
\item  $c'_t > r+ \mu^{Cost}_{BP}$
\item  $(c'_t-r-\mu^{Cost}_{BP}) \cdot z \geq \mu^{Cost}_{CM}$
\item $c'_t$ belongs to the $\min \{m \cdot c_{Incl}, size_{L_{CM,c_{block}}}\}$ highest bids
\end{itemize}
then this transaction will be included in both an inclusion list and the block.
\end{proof}
\paragraph{Censorship resistant}

Single TFM is \textbf{censorship resistant} because the allocation rules of all the types of  includers and the block producer ignore bribes.
\subsubsection{Intuition about How the Choice $z,c_{t_0}$ Affect the Minimum Bribe Needed for Censorship}
\label{app:zminimumbribe}
Based on Theorems \ref{theoremA}, \ref{theoremC} we prove in Section \ref{simplified} for conditional inclusion lists, assuming that $\mu^{Cost}_{CM}=\mu^{Cost}_{BP}=0$ and $w=c_{block}$, it holds that: if a briber gives a sum of bribes higher than  $min \{m \cdot f_{CM}, r \cdot s \} +f_{BP}$ , there is no Nash equilibrium where the target transaction $t_0$ is included in the block. In this section, we examine how this amount is affected by the choice of $c_{t_0}$ or $z$. In the Single TFM, this amount is equal to

\begin{align} & min \{m \cdot f_{CM}, r \cdot s \} +f_{BP} = \\ & min \{m \cdot(c_{t_0}-r)\cdot s \cdot   z, r \cdot s \} + \\ &(c_{t_0}-r) \cdot (1-z)\cdot s \end{align}

Note that  $f_{CM}= (c_{t_0}-r) \cdot s \cdot  z$ and $f_{BP}= (c_{t_0}-r) \cdot (1-z)$, because $t_0$ belongs to first $\min \{m \cdot c_{Incl}, size_{L_{CM,c_{block}}}\}$ positions in  $L_{CM,c_{block}}$ and thus it holds  $c_{t_0} \geq r$. \par  

If we fix $z$, and the mempool $M$ apart from the target transaction, then the higher the $c_{t_0}$, the higher the above amount is. \par 

Now we fix $M,c_{t_0},r$ and we examine for which $z$ the above formula is maximised. We prove that it is maximised for $z_0= \min \{\frac{\frac{r \cdot s}{m}}{(c_{t_0}-r)},1\}$. \par  

Note that:
\begin{itemize}

\item The latter amount depends on $r,c_{t_0}$, which means that it does not apply to every target transaction.
\item  It holds that the higher $r$, the higher $z_0$. The intuition about this is the following: when $r$ is very high  $min \{m \cdot f_{CM}, r \cdot s \} = m \cdot f_{CM}$, which means that equations $2,3$ depend not only on the block producer fee but also on the committee fee.
\end{itemize}
\begin{proof} We have the following two cases:
\begin{itemize}
\item  $\frac{r }{m}  > c_{t_0}-r$. In this case, it holds  $0 \leq z \leq \min \{1, \frac{\frac{r }{m}}{(c_{t_0}-r)}\}$, because it holds $z \leq 1$.
\item $\frac{r }{m} \leq c_{t_0}-r$. In this case, we have either:
\begin{itemize}
    \item  $0 \leq z \leq  \frac{\frac{r }{m}}{(c_{t_0}-r)}$, or
    \item  $\frac{\frac{r }{m}}{(c_{t_0}-r)}< z \leq 1$.
\end{itemize}
\end{itemize}
Thus,
\begin{itemize}
\item  When $\frac{r }{m} \leq c_{t_0}-r \text{ and }1 \geq z >  \frac{\frac{r }{m}}{(c_{t_0}-r)}$, the equation $2,3$  is equal to

$$
r \cdot s+ (c_{t_0}-r) \cdot (1-z)\cdot s
$$

This is maximised for $z_0= \frac{\frac{r }{m}}{(c_{t_0}-r)}$.

\item  When $[\frac{r }{m}  > c_{t_0}-r]$  OR $[\frac{r }{m}\leq c_{t_0}-r\text{ and }0 \leq z \leq  \frac{\frac{r }{m}}{(c_{t_0}-r)}]$ the equation $2,3$ is equal to

$$
m \cdot(c_{t_0}-r)\cdot s \cdot   z +(c_{t_0}-r) \cdot (1-z)\cdot s
$$
\end{itemize}
The above formula is maximised for $z_0= \min \{\frac{\frac{r }{m}}{(c_{t_0}-r)},1\}$.
\end{proof}

\section{Excluding Strategies that Add Fake Transactions to the Mempool}
\label{simplified}
In this section, we restrict the strategy space of both includers and the block producer by excluding strategies that involve adding fake transactions to the mempool. However, we still allow strategies that add fake transactions directly to the block or inclusion list. We analyse how the minimum bribe required by a briber to censor a transaction varies based on whether the inclusion lists are conditional or unconditional. 
\begin{itemize}
\item  Let $M$ be an arbitrary mempool and $t_0$ be the target transaction the briber tries to censor.
\item  Let $s$ be the size of every transaction in the mempool $M$.
\item  $r\cdot s$ the burning fee per transaction.
\item  Let $c_{block}, c_{Incl}$ be the maximum number of transactions the block and the inclusion list can store respectively.
\item Let $sum$ be the block rewards (block producer fees minus costs).
\item  Let $w$ be the number of transactions in $M$ (note that $M=M_0$ because theer are not afke transactions in the mempool by assumption.
\item  Let $f_{CM}$ be the fee the committee receives if $t_0$ is included in an inclusion list and the block.
\item  Let $f_{BP}$ be the fee the block producer receives if  $t_0$ is included in the block regardless of whether this transaction is included in an inclusion list.
\item  Let us assume $\mu^{Cost}_{BP}=\mu^{Cost}_{CM}=0$.
\item The includers and the block producer have only one type that known to the other parties. Thus, we use the notion of Nash equilibrium instead of Bayesian Nash equilibrium.
\item  The bribe function of the block producer offers them bribe $B_1$ when they omit transaction $t_0$ from their block.
\item The bribe function of a includer $j$ gives them bribe $B^j$ when they omit transaction $t_0$ from their block.
\end{itemize}
\subsubsection{Conditional Lists}

\begin{theorem}\label{theoremA}  Regardless of the payment rule, if  $B_1>f_{BP} + r \cdot s \cdot (\max \{c_{block}-w+1,0\})$, it is a dominant strategy for the block producer to censor $t_0$, which means that there is no Nash equilibrium where $t_0$ is included in the block.
\end{theorem}
\begin{proof} This holds because regardless of whether $t_0$ has been included in an inclusion list, the block producer can omit $t_0$ without being rejected by the attesters by adding fake transactions to their block. The maximum cost the block producer incurs when they omit $t_0$ is $f_{BP}+r \cdot s \cdot (\max \{c_{block}-w+1,0\})$. Thus, it is more profitable for them to receive the bribe and omit $t_0$.
\end{proof}
\begin{theorem} Regardless of the payment rule, if $B_1>f_{BP}$, there is a Nash equilibrium where both the includers and the block producer censor $t_0$. 
\end{theorem}
\begin{proof} This holds because: (i) when the block producer omits $t_0$ then the committee does not receive $f_{CM}$ even if they include $t_0$ in an inclusion list and (ii) when $t_0$ is not in an inclusion list, the maximum cost  of the block producer to omit it is $f_{BP}$.
\end{proof}
\begin{theorem}\label{theoremC} Regardless of the payment rule, if $B_1>f_{BP}$ and  $\forall j \in \{ 1, \ldots m \} :B^j >f_{CM}$ there is no Nash equilibrium where $t_0$ is included in the block.
\end{theorem}
\begin{proof} We will prove it by contradiction. Let us assume that there exist such a Nash equilibrium. This means that in this strategy profile the block producer has included $t_0$ in the block. Moreover, as  $\forall j \in \{ 1, \ldots m \} :B^j >f_{CM}$ , it holds that no includer has included $t_0$ in their inclusion list; otherwise they could increase their utility by omitting $t_0$ and receiving the bribe. Note that regardless of the payment rule, the reward of an includer who includes $t_0$ cannot be more than $f_{CM}$. As $t_0$ is not in an inclusion list, the block producer can omit it and incur maximum cost $f_{BP}$. We reach a contradiction because the block producer can increase their utility by omitting it.
\end{proof}
\subsubsection{Unconditional Lists}

\begin{theorem}  When $m \cdot c_{Incl} \leq c_{block}$, regardless of the payment rule, if $B_1< sum$, it is a dominant strategy for the block producer to follow an allocation rule that includes $t_0$ if it is included in an inclusion list. \end{theorem}

\begin{proof}This holds because if the block producer ignores any transaction from an inclusion list, they will lose their block rewards. Recall that in this section we have excluded strategies where the block producer can add fake transactions to the mempool to capture space in the inclusion lists, thereby allowing them to omit $t_0$ at a lower cost.
\end{proof}
\begin{theorem} 
Assume $m \cdot w \leq m \cdot c_{Incl} \leq c_{block}$. Under the payment rule where the committee fee is awarded to the includer with the smallest order who includes it, if ($B_1< sum$ AND $\sum_{j \in \{1, \ldots m\}} \cdot B^j <m\cdot  f_{CM}$), there is no Nash equilibrium where $t_0$ is not included in the block. 
\end{theorem}

\begin{proof} We will prove it by contradiction. Assume that such a Nash equilibrium exists. As it holds $B_1< sum$, we know by the previous theorem that at a Nash equilibrium, the block producer follows an allocation rule that includes $t_0$ if it is included in an inclusion list. As $t_0$ has not been included in the block, this means that no includer has included $t_0$ in their inclusion list.  As it holds  $\sum_{j \in \{1, \ldots m\}} \cdot B^j <m\cdot  f_{CM}$, there is at least one includer $j$ with $B^j<f_{CM}$. We reach a contradiction because this includer can increase their utility by deviating and including $t_0$ in their inclusion list, as they know that this transaction will be included in the block and will give them committee fee $f_{CM}$. 
\end{proof}

Note that in \cite{cryptoeprint:2025/194}, the authors propose a notion for censorship resistance based on the minimum amount required to censor a transaction.

\subsubsection{Conclusion}

Based on the above theorems, we conclude that unconditional lists can significantly  increase the cost of bribing at a Nash equilibrium provided we exclude strategies where the block producer can add fake transactions to the mempool. However, if we include the latter strategy in the strategy space, the cost of bribing between conditional and unconditional lists becomes similar. This holds because, in both cases, adding fake transactions to the mempool and later invalidate them is usually the lowest-cost strategy that allows the block producer to omit a transaction from an inclusion list.

\end{document}